\documentclass[preprint]{sigplanconf}
\usepackage{proof}
\usepackage{amssymb}
\usepackage{amsmath}
\usepackage{txfonts}
\usepackage{graphicx}
\newtheorem{lemma}{Lemma}[section]

\def\lshfit#1#2{\kern-#1 #2\kern#1}

\def\fillsquare{\kern2pt\raise0.25pt
     \hbox{$\vcenter{\hrule height0pt \hbox{\vrule width5pt height5pt} \hrule height0pt}$}}
\newenvironment{proof}{%
\vskip6pt{\em Proof\kern6pt}}{%
{\unskip\nobreak\hfill\penalty50\kern4pt\hbox{}\nobreak\hfill\fillsquare}\vskip6pt}
\begin{document}
\baselineskip=9.750pt

\setlength{\pdfpagewidth}{\paperwidth}
\setlength{\pdfpageheight}{\paperheight}

\conferenceinfo{}{}
\copyrightyear{2015} 
\copyrightdata{978-1-nnnn-nnnn-n/yy/mm} 
\doi{nnnnnnn.nnnnnnn}

\title{%
Linearly Typed Dyadic Group Sessions
\break for Building Multiparty Sessions
}

\authorinfo%
{Hongwei Xi\and Hanwen Wu}{Boston University}{\{hwxi,hwwu\}@cs.bu.edu}
\maketitle

\def\MTLC{\mbox{MTLC}}

\begin%
{abstract}
Traditionally, each party in a (dyadic or multiparty) session
implements exactly one role specified in the type of the session. We
refer to this kind of session as an individual session (i-session). As
a generalization of i-session, a group session (g-session) is one in
which each party may implement a group of roles based on one
channel. In particular, each of the two parties involved in a dyadic
g-session implements either a group of roles or its complement.  In
this paper, we present a formalization of g-sessions in a
multi-threaded lambda-calculus (MTLC) equipped with a linear type
system, establishing for the MTLC both type preservation and global
progress.  As this formulated MTLC can be readily embedded into ATS, a
full-fledged language with a functional programming core that supports
both dependent types (of DML-style) and linear types, we obtain a
direct implementation of linearly typed g-sessions in ATS. The primary
contribution of the paper lies in both of the identification of
g-sessions as a fundamental building block for multiparty sessions and
the theoretical development in support of this identification.

\end{abstract}
\def\simp{\rightarrow}
\def\timp{\rightarrow}
\def\ctimp{\Rightarrow}
\def\itimp{\rightarrow_{i}}
\def\ltimp{\rightarrow_{l}}
\def\cres{rc}
\def\const{c}
\def\cfun{\mbox{\it cf}}
\def\ccon{\mbox{\it cc}}
\def\ctrue{\mbox{\it true}}
\def\cfalse{\mbox{\it false}}
\def\exp{e}
\def\vexp{\vec{\exp}}
\def\val{v}
\def\vval{\vec{\val}}
\def\xf{\mbox{\it x{\kern0.5pt}f}}
\def\dif{\mbox{\tt if}}
\def\dfst#1{\mbox{\tt fst}(#1)}
\def\dsnd#1{\mbox{\tt snd}(#1)}
\def\dunit{\langle\rangle}
\def\tuple#1{\langle#1\rangle}
\def\lam#1#2{{\tt lam}\;#1.\,#2}
\def\app#1#2{{\tt app}(#1, #2)}
\def\fix#1#2{{\tt fix}\;#1.\,#2}
\def\letin#1#2{{\tt let}\;#1\;{\tt in}\;#2\;{\tt end}}
\def\ty{T}
\def\vw{V}
\def\vwty{\hat{T}}
\def\tvar{\alpha}
\def\vtvar{\hat{\alpha}}
\def\tunit{\mbox{\bf 1}}
\def\tbase{\delta}
\def\tint{\mbox{\bf int}}
\def\tbool{\mbox{\bf bool}}
\def\vtbase{\hat{\tbase}}
\def\tpjg{\vdash}
\def\temd{\models}
\def\SIG{\mbox{\rm SIG}}
\def\langz{\MTLC_0}
\def\langch{\MTLC_{{\rm ch}}}
\def\ST{{\it S}}
\def\stvar{\sigma}
\def\stmsg{\mbox{\tt msg}}
\def\stnil{\mbox{\tt nil}}
\def\stnild{\overline{\stnil}}
\def\dual#1{\mbox{\it dual}(#1)}
\def\chmsg#1#2{\mbox{\tt msg}(#1)::#2}
\def\fsend{\underline{\mbox{\it send}}}
\def\frecv{\underline{\mbox{\it recv}}}
\def\fskip{\underline{\mbox{\it skip}}}
\def\fclose{\underline{\mbox{\it close}}}
\def\fchancreate{\mbox{\it chan\_create}}
\def\fchanlinkcreate{\mbox{\it chan2\_link\_create}}
\def\fchanlinktwo{\mbox{\it chan2\_link}}
\def\fchanlinkthree{\mbox{\it chan3\_link}}
\def\fchposneglink{\mbox{\it chposneg\_link}}
\def\chcst{\mbox{\it ch}}
\def\chpcst{\mbox{\it ch}^{+}}
\def\chncst{\mbox{\it ch}^{-}}
\def\fchancreatetwo{\mbox{\it chan2\_create}}
\def\iset#1{\{#1\}}
\section{Introduction}
\label{section:introduction}
In broad terms, a session is a sequence of interactions between two or
more concurrently running programs (often referred to as parties), and
a session type is a form of type for specifying (or classifying)
sessions. Traditionally, each party in a session implements exactly
one role in the session type assigned to the session. For instance,
each of the two parties in a dyadic session implements either the role
of a client or the role of a server.  Let us suppose that there are
more than two roles in a session type (e.g., seller, buyer 1, and
buyer 2). Conceptually, we can assign this session type to a session
in which one party may implement a group of roles. For instance, there
may be two parties in the session such that one implements the role of
seller and the other implements both of the roles of buyer~1 and
buyer~2. We coin the name {\em g-session} (for group session) to refer
to a session in which a party may implement multiple roles. In
contrast, a session is referred to as an {\em i-session} (for
individual session) if each party in the session implements exactly
one role. Therefore, an i-session is just a special case of g-session
where each involved group is a singleton. As far as we can tell,
this form of generalization from (dyadic) i-sessions to (dyadic)
g-sessions is novel.

\def\chsnd#1#2{\mbox{\tt snd}(#1)::#2}
\def\chrcv#1#2{\mbox{\tt rcv}(#1)::#2}
As an example (for clarifying basic concepts), let us assume that a
dyadic session consists of two running programs (parties) P and Q that
are connected with a bidirectional channel. From the perspective of P,
the channel (that is, the endpoint at P's side)
may be specified by a term sequence of the following form:
$$\chsnd{\tint}{\chsnd{\tint}{\chrcv{\tbool}{\stnil}}}$$ which means
that an integer is to be sent, another integer is to be sent, a
boolean is to be received, and finally the channel is to be closed.
Clearly, from the perspective of Q, the channel (that is, the endpoint
at Q's side) should be specified by the following term sequence:
$$\chrcv{\tint}{\chrcv{\tint}{\chsnd{\tbool}{\stnil}}}$$ which means
precisely the dual of what the previous term sequence does.  We may think
of P as a client who sends two integers to the server Q and then receives
from Q either true or false depending on whether or not the first sent
integer is less than the second one. A simple but crucial observation
is that the above two term sequences can be unified as follows:
$$\chmsg{0,1,\tint}{\chmsg{0,1,\tint}{\chmsg{1,0,\tbool}{\stnil}}}$$
where $0$ and $1$ refer to the two roles implemented by P and Q,
respectively. Given a type $T$, $\stmsg(i,j,T)$ means a value of the
type $T$ is transferred from the party implementing role $i$ to the
one implementing role $j$, where both $i$ and $j$ range over $0$ and
$1$.

\def\tchan{\mbox{\bf chan}}
\begin%
{figure}
\fontsize{8pt}{9pt}\selectfont
\begin%
{verbatim}
fun P() = let
  val () =
    channel_send(CH, I1, 0, 1) // send to Q
  val () =
    channel_send(CH, I2, 0, 1) // send to Q
  val b0 = channel_recv(CH, 1, 0) // recv from Q
  val () = channel_close(CH) // close the P-end of CH
in b0 end (* end of [P] *)

fun Q() = let
  val i1 =
    channel_recv(CH, 0, 1) // recv from P
  val i2 =
    channel_recv(CH, 0, 1) // recv from P
  val () =
    channel_send(CH, i1 < i2, 1, 0) // send to P
  val () = channel_close(CH) // close the Q-end of CH
in () end (* end of [Q] *)
\end{verbatim}
\caption{Some pseudo code in ML-like syntax}
\label{figure:P-and-Q}
\end{figure}
In Figure~\ref{figure:P-and-Q},
we present some pseudo code showing a plausible way to implement the
programs P and Q. Please note that the functions P and Q, though
written together here, can be written in separate contexts. We use
$\mbox{CH}$ to refer to a channel available in the surrounding context
of the code and $\mbox{I1}$ and $\mbox{I2}$ for two integers; the
functions $\mbox{\tt channel\_send}$ and $\mbox{\tt channel\_recv}$
are for sending and receiving data via a given channel, and $\mbox{\tt
  channel\_close}$ for closing a given channel.

Let us now sketch a way to make the above pseudo code typecheck.
Given an integer $i$ and a session type $\ST$, let $\tchan(i,\ST)$ be
the type for a channel of role $i$, that is, a channel held by a party
for implementing role $i$.
We can assign the following type to $\mbox{\tt channel\_send}$:
$$(!\tchan(i,\chmsg{i,j,\ty}{\ST}) \gg \tchan(i,\ST), \tint(i), \tint(j),
\ty) \timp \tunit$$ where $i\neq j$ is assumed, and $\tint(i)$ and $\tint(j)$
are singleton types for integers equal to $i$ and $j$, respectively, and $\ty$
and $\ST$ stand for a type and a session type, respectively.
Basically, this type%
\footnote{%
Strictly speaking, this type should be referred to as a type schema as
it contains occurrences of meta-variables.%
} means that calling $\mbox{\tt channel\_send}$ on a channel of the
type $\tchan(i,\chmsg{i,j,\ty}{\ST})$, integer $i$, integer $j$ and a
value of the type $\ty$ returns a unit while {\it changing} the type
of the channel to $\tchan(i,\ST)$.  Clearly, $\tchan$ is required be a
linear type constructor for this to make sense. As can be expected,
the type assigned to $\mbox{\tt channel\_recv}$ should be of the
following form:
$$(!\tchan(j,\chmsg{i,j,\ty}{\ST}) \gg \tchan(j,\ST), \tint(i),
\tint(j)) \timp \ty$$ where $i\neq j$ is assumed. This type
essentially indicates that calling $\mbox{\tt channel\_recv}$ on a
channel of the type $\tchan(j,\chmsg{i,j,\ty}{\ST})$, integer $i$ and
integer $j$ returns a value of the type $\ty$ while {\it changing} the
type of the channel to $\tchan(j,\ST)$.  As for $\mbox{\tt
  channel\_close}$, it is assigned the following
type: $$(\tchan(i,\stnil))\timp\tunit$$ indicating that calling
$\mbox{\tt channel\_close}$ on a channel consumes the channel (so that
the channel is no longer available for use).

\def\setcomp#1{{\overline#1}}
Given an integer $i$ (representing a role) and a session type $\ST$,
the type $\tchan(i, \ST)$ for single-role channels can be naturally
transitioned into one of the form $\tchan(G, \ST)$ for multirole
channels, where $G$ stands for a finite set of integers (representing
roles). The fundamental issue to be addressed in this transition is to
figure out a consistent interpretation for each term $\stmsg(i,j,T)$
by a party based on the group of roles it implements.  Assume there
exists a fixed set of $n$ roles ranging from $0$ to $n-1$ for some
$n\geq 2$. For each $G$, we use $\setcomp{G}$ for the complement of
$G$, which consists of all of the natural numbers less than $n$ that
are not in $G$.  We have the following four scenarios for interpreting
$\stmsg(i,j,T)$ based on the group $G$ of roles implemented by a party:
\begin%
{itemize}
\item
Assume $i\in G$ and $j\in G$.
Then $\stmsg(i,j,T)$ is interpreted as an internal message, and
it is ignored.
\item
Assume $i\in G$ and $j\not\in G$.  Then $\stmsg(i,j,T)$ is interpreted
as sending a value of the type $T$ by the party implementing $G$ to
the party implementing $\setcomp{G}$.
\item
Assume $i\not\in G$ and $j\in G$.  Then $\stmsg(i,j,T)$ is interpreted
as receiving a value of the type $T$ by the party implementing $G$
from the party implementing $\setcomp{G}$.
\item
Assume $i\not\in G$ and $j\not\in G$.  Then $\stmsg(i,j,T)$ is
interpreted as an external message, and it is ignored.
\end{itemize}
With this interpretation,
$\mbox{\tt channel\_send}$ can be assigned the following type:
$$(!\tchan(G,\chmsg{i,j,\ty}{\ST}) \gg \tchan(G,\ST),\tint(i),\tint(j),
\ty) \timp \tunit$$
where $i\in G$ and $j\not\in G$ is assumed;
$\mbox{\tt channel\_recv}$ can be assigned the following type:
$$(!\tchan(G,\chmsg{i,j,\ty}{\ST}) \gg \tchan(G,\ST),\tint(i),\tint(j))
\timp \ty$$
where $i\not\in G$ and $j\in G$ is assumed. As for $\mbox{\tt channel\_close}$,
the following type is assigned:
$$(\tchan(G,\stnil))\timp\tunit$$
While transitioning single-role channels into multirole channels may
seem mostly intuitive, there are surprises. In particular, we have the
following result for justifying the use of multirole channels as a
building block for implementing multiparty sessions (that involve more
than 2 parties):
\begin%
{theorem}
\label{theorem:chan2_link_intro}
Assume that $\chcst_0$ and $\chcst_1$ are two multirole channels
(held by a party belonging to two sessions) of the types
$\tchan(G_0,\ST)$ and $\tchan(G_1,\ST)$, respectively, where
$\setcomp{G_0}$ and $\setcomp{G_1}$ are disjoint. Then there is a
generic method for building a multirole channel $\chcst_2$ of the type
$\tchan(G_0\cap G_1,\ST)$ such that each message received on one of
$\chcst_0$, $\chcst_1$ and $\chcst'_2$ can be forwarded onto one of
the other two in a type-correct manner, where $\chcst'_2$ refers to
the dual of $\chcst_2$.
\end{theorem}
The significance of Theorem~\ref{theorem:chan2_link_intro} will be
elaborated later on. Intuitively, this theorem justifies some form of
``wiring'' to allow two existing channels to be connected to provide
the behavior of new channel (related to them in some way), enabling
a multiparty session to be built based on dyadic g-sessions.

ATS~\cite{ATS-types03,CPwTP} is a full-fledged language with a
functional programming core based on ML that supports both dependent
types (of DML-style~\cite{XP99,DML-jfp07}) and linear types. Its
highly expressive type system makes it largely straightforward to
implement session types in ATS (e.g., based on the outline given
above) if our concern is primarily about type-correctness. For
instance, there have already been implementations of session types in
Haskell (e.g.,~\cite{NeubauerT04,PucellaT08}) and elsewhere that offer
type-correctness. However, mere type-correctness is inadequate. We are
to establish formally the property that concurrency based on session
types (formulated in this paper) can never result in deadlocking,
which is often referred to as {\em global progress}.  There have been
many formalizations of session types in the literature (e.g.,
~\cite{Honda93,HondaVK98,CastagnaDGP09,GayV10,Vasconcelos12,Wadler12,ToninhoCP13}).
Often the dynamics formulated in a formalization of session types is
based on $\pi$-calculus~\cite{PiCalculus} or its variants/likes.  We
instead use multi-threaded $\lambda$-calculus (MTLC) as a basis for
formalizing session types as such a formalization is particularly
suitable for guiding implementation (due to its being less abstract
and more operational).

The rest of the paper is organized as follows.  In
Section~\ref{section:langz}, we formulate a multi-threaded
$\lambda$-calculus $\langz$ equipped with a simple linear type system,
setting up the basic machinery for further development. We then extend
$\langz$ to $\langch$ in Section~\ref{section:langch} with support for
session types and establish both type preservation and global progress
for $\langch$. In Section~\ref{section:dyadic_to_multiparty}, We
establish a key theorem needed for building multiparty sessions based
on dyadic sessions. We present a few commonly used constructors for
session types in Section~\ref{section:more_st_constructors} and then
briefly mention the implementation of a classic example of 3-party
sessions in Section~\ref{section:S0B1B2}. We also mention some key
steps taken in both of our implementations of session-typed channels
in ATS and in Erlang in Section~\ref{section:implementation}.
Lastly, we discuss some closely related work in
Section~\ref{section:related-conclusion} and then conclude.

The primary contribution of the paper lies in both of the
identification of g-sessions as a fundamental building block for
multiparty sessions and the theoretical development in support of this
identification.  We consider the formulation and proof of
Theorem~\ref{theorem:chan2_link_intro} a particularly important part
of this contribution.

\section%
{$\langz$ with Linear Types}
\label{section:langz}
\begin%
{figure}
\[%
\fontsize{8pt}{9pt}\selectfont
\begin%
{array}{lrcl}
\mbox{expr.} & \exp & ::= &
x \mid f \mid \cres \mid \const(\vexp) \mid \\
             &      &     &
\dunit \mid \tuple{\exp_1,\exp_2} \mid \dfst{\exp} \mid \dsnd{\exp} \mid \\
             &      &     &
\letin{\tuple{x_1,x_2}=\exp_1}{\exp_2} \mid \\
             &      &     &
\lam{x}{\exp} \mid \app{\exp_1}{\exp_2} \mid \fix{f}{\val}  \\
\mbox{values} & v & ::= &
x \mid \cres \mid \ccon(\vval) \mid \dunit \mid \tuple{\val_1,\val_2} \mid \lam{x}{\exp} \\
\mbox{types} & \ty & ::= &
\tbase \mid \tunit \mid \ty_1*\ty_2 \mid \vwty_1\itimp\vwty_2  \\
\mbox{viewtypes} & \vwty & ::= & \vtbase \mid \ty \mid \vwty_1\otimes\vwty_2 \mid \vwty_1\ltimp\vwty_2 \\
\mbox{int. expr. ctx.\kern-6pt} & \Gamma & ::= & \emptyset \mid \Gamma, \xf:\ty \\
\mbox{lin. expr. ctx.\kern-6pt} & \Delta & ::= & \emptyset \mid \Delta, x:\vwty \\
\end{array}\]
\caption{Some syntax for $\langz$}
\label{figure:langz:syntax}
\end{figure}
We first present a multi-threaded lambda-calculus $\langz$ equipped with a
simple linear type system, setting up the basic machinery for further
development. The dynamic semantics of $\langz$ can essentially be seen as
an abstract form of evaluation of multi-threaded programs.

Some syntax of $\langz$ is given in Figure~\ref{figure:langz:syntax}.
We use $x$ for a lam-variable and $f$ for a fix-variable, and $\xf$ for
either a lam-variable or a fix-variable. Note that a lam-variable is
considered a value but a fix-variable is not.
We use $\cres$ for constant resources and $\const$ for constants, which
include both constant functions $\cfun$ and constant constructors $\ccon$.
We treat resources in $\langz$ abstractly and will later introduce
communication channels as a concrete form of resources.
The meaning of various standard forms of expressions in $\langz$
should be intuitively clear.  We may refer to a closed expression
(containing no free variables) as a {\em program}.

We use $\ty$ and $\vwty$ for (non-linear) types and (linear)
viewtypes, respectively, and refer $\vwty$ to as a true viewtype if it
is a viewtype but not a type. We use $\tbase$ and $\vtbase$ for base
types and base viewtypes, respectively.  For instance, $\tbool$ is the
base type for booleans and $\tint$ for integers.  We also assume the
availability of integer constants when forming types. For instance, we
may have a type $\tint(i)$ for each integer constant $i$, which can
only be assigned to a (dynamic) value equal to integer $i$.

For a simplified presentation, we do not introduce any concrete base
viewtypes in $\langz$. We assume a signature $\SIG$ for assigning a
viewtype to each constant resource $\cres$ and a constant type
(c-type) schema of the form $(\vwty_1,\ldots,\vwty_n)\ctimp\vwty$ to
each constant. For instance, we may have a constant function
$\mbox{\it iadd}$ of the following c-type schema:
$$(\tint(i), \tint(j))\timp\tint(i+j)$$
where $i$ and $j$ are meta-variables ranging over integer constants;
each occurrence of $\mbox{\it iadd}$ in a program is given a c-type that
is an instance of the c-type schema assigned to $\mbox{\it iadd}$.


Note that a type is always considered a viewtype.  Let $\vwty_1$ and
$\vwty_2$ be two viewtypes. The type constructor $\otimes$ is based on
multiplicative conjunction in linear logic.  Intuitively, if a
resource is assigned the viewtype $\vwty_1\otimes\vwty_2$, then the
resource is a conjunction of two resources of viewtypes $\vwty_1$ and
$\vwty_2$.  The type constructor $\ltimp$ is essentially based on
linear implication $\multimap$ in linear logic.  Given a function of
the viewtype $\vwty_1\ltimp\vwty_2$ and a value of the viewtype
$\vwty_1$, applying the function to the value yields a result of the
viewtype $\vwty_2$ while the function itself is consumed. If the
function is of the type $\vwty_1\itimp\vwty_2$, then applying the
function does not consume it. The subscript $i$ in $\itimp$ is often
dropped, that is, $\timp$ is assumed to be $\itimp$ by default.  The
meaning of various forms of types and viewtypes is to be made clear
and precise when the rules are presented for assigning viewtypes to
expressions in $\langz$.

\def\fthreadcreate{\mbox{\it thread\_create}}
There is a special constant function $\fthreadcreate$ in $\langz$
for thread creation, which is assigned the following interesting
c-type:
\[\begin{array}{rcl}
\fthreadcreate & : & (\tunit\ltimp\tunit) \ctimp \tunit
\end{array}\]
A function of the type $\tunit\ltimp\tunit$ is a procedure that takes no
arguments and returns no result (when its evaluation terminates).  Given
that $\tunit\ltimp\tunit$ is a true viewtype, a procedure of this type may
contain resources and thus must be called exactly once. The operational
semantics of $\fthreadcreate$ is to be formally defined later.

\def\dom{\mbox{\bf dom}}
\def\mapadd#1#2#3{#1[#2\mapsto #3]}
\def\mapdel#1#2{#1\backslash#2}
\def\maprep#1#2#3{#1[#2 := #3]}
A variety of mappings, finite or infinite, are to be introduced in the
rest of the presentation.  We use $[]$ for the empty mapping and
$[i_1,\ldots,i_n\mapsto o_1,\ldots,o_n]$ for the finite mapping that maps
$i_k$ to $o_k$ for $1\leq k\leq n$.  Given a mapping $m$, we write
$\dom(m)$ for the domain of $m$. If $i\not\in\dom(m)$, we use
$\mapadd{m}{i}{o}$ for the mapping that extends $m$ with a link from $i$ to
$o$.  If $i\in\dom(m)$, we use $\mapdel{m}{i}$ for the mapping obtained
from removing the link from $i$ to $m(i)$ in $m$, and $\maprep{m}{i}{o}$
for $\mapadd{(\mapdel{m}{i})}{i}{o}$, that is, the mapping obtained from
replacing the link from $i$ to $m(i)$ in $m$ with another link from $i$ to
$o$.

\def\resof{\rho}
\begin{figure}
\[\begin{array}{rcl}
\resof(\cres) & = & \{\cres\} \\
\resof(\const(\exp_1,\ldots,\exp_n)) & = & \resof(\exp_1)\uplus\cdots\uplus\resof(\exp_n) \\
\resof(\xf) & = & \emptyset \\
\resof(\dunit) & = & \emptyset \\
\resof(\tuple{\exp_1,\exp_2}) & = & \resof(\exp_1)\uplus\resof(\exp_2) \\
\resof(\dfst{\exp}) & = & \resof(\exp) \\
\resof(\dsnd{\exp}) & = & \resof(\exp) \\
\resof(\dif(\exp_0,\exp_1,\exp_2)) & = & \resof(\exp_0)\uplus\resof(\exp_1) \\
\resof(\letin{\tuple{x_1,x_2}=\exp_1}{\exp_2}) & = & \resof(\exp_1)\uplus\resof(\exp_2) \\
\resof(\lam{x}{\exp}) & = & \resof(\exp) \\
\resof(\app{\exp_1}{\exp_2}) & = & \resof(\exp_1)\uplus\resof(\exp_2) \\
\resof(\fix{f}{\val}) & = & \resof(\val) \\
\end{array}\]
\caption{The definition of $\resof(\cdot)$}
\label{figure:resof}
\end{figure}
We define a function $\resof(\cdot)$ in Figure~\ref{figure:resof} to
compute the multiset (that is, bag) of constant resources in a given
expression.  Note that $\uplus$ denotes the multiset union. In the type
system of $\langz$, it is to be guaranteed that $\resof(\exp_1)$ equals
$\resof(\exp_2)$ whenever an expression of the form
$\dif(\exp_0,\exp_1,\exp_2)$ is constructed, and this justifies
$\resof(\dif(\exp_0,\exp_1,\exp_2))$ being defined as
$\resof(\exp_0)\uplus\resof(\exp_1)$.

\def\rns{R}
\def\RES{{\bf RES}}
We use $\rns$ to range over finite multisets of resources. Therefore,
$\rns$ can also be regarded as a mapping from resources to natural numbers:
$\rns(\cres) = n$ means that there are $n$ occurrences of $\cres$ in
$\rns$. It is clear that we may not combine resources arbitrarily. For
instance, we may want to exclude the combination of one resource stating
integer 0 at a location L and another one stating integer 1 at the same
location. We fix an abstract collection $\RES$ of finite multisets of
resources and assume the following:
\begin{itemize}
\item
$\emptyset\in\RES$.
\item
For any $\rns_1$ and
$\rns_2$, $\rns_2\in\RES$ if $\rns_1\in\RES$ and $\rns_2\subseteq\rns_1$,
where $\subseteq$ is the subset relation on multisets.
\end{itemize}
We say that $\rns$ is a valid multiset of resources if $\rns\in\RES$ holds.

\def\pool{\Pi}
\def\tid{\mbox{\it tid}}
\def\esubst{\theta}
\def\tsubst{\Theta}

In order to formalize threads, we introduce a notion of {\em
  pools}. Conceptually, a pool is just a collection of programs (that is,
closed expressions).  We use $\pool$ for pools, which are formally defined
as finite mappings from thread ids (represented as natural numbers) to
(closed) expressions in $\langz$ such that $0$ is always in the domain of
such mappings.  Given a pool $\pool$ and $\tid\in\dom(\pool)$, we refer to
$\pool(\tid)$ as a thread in $\pool$ whose id equals $\tid$. In particular,
we refer to $\pool(0)$ as the main thread in $\pool$.  The definition of
$\resof(\cdot)$ is extended as follows to compute the multiset of resources
in a given pool:
\begin%
{center}
$\resof(\pool)=\biguplus_{\tid\in\dom(\pool)}\resof(\pool(\tid))$
\end{center}
We are to define a relation on pools in
Section~\ref{section:langz:dynamics} to simulate multi-threaded program
execution.

\begin{figure}
\fontsize{8}{8}\selectfont
\[\begin{array}{c}
\infer[\mbox{\bf (ty-res)}]
      {\Gamma;\emptyset\tpjg\cres:\vtbase}
      {\SIG\temd\cres:\vtbase} \\[2pt]
\infer[\mbox{\bf (ty-cst)}]
      {\Gamma;\Delta_1,\ldots,\Delta_n\tpjg\const(\exp_1,\ldots,\exp_n):\vwty}
      {$$\begin{array}{c}
       \SIG\temd\const:(\vwty_1,\ldots,\vwty_n)\ctimp\vwty \\
       \Gamma;\Delta_i\tpjg\exp_i:\vwty_i~~\mbox{for $1\leq i\leq n$} \\
       \end{array}$$} \\[2pt]
\infer[\mbox{\bf (ty-var-i)}]
      {(\Gamma,\xf:\ty;\emptyset)\tpjg \xf:\ty}
      {} \\[2pt]
\infer[\mbox{\bf (ty-var-l)}]
      {(\Gamma;\emptyset,x:\vwty)\tpjg x:\vwty}
      {} \\[2pt]
\infer[\mbox{\bf(ty-if)}]
      {\Gamma;\Delta_0,\Delta\tpjg\dif(\exp_0,\exp_1,\exp_2):\vwty}
      {$$\begin{array}{c}
       \Gamma;\Delta_0\tpjg\exp_0:\tbool \\
       \Gamma;\Delta\tpjg \exp_1:\vwty \kern6pt
       \Gamma;\Delta\tpjg \exp_2:\vwty \kern6pt
       \resof(\exp_1) = \resof(\exp_2) \\
       \end{array}$$} \\[2pt]
\infer[\mbox{\bf(ty-unit)}]
      {\Gamma;\emptyset\tpjg\dunit:\tunit}{} \\[2pt]
\infer[\mbox{\bf(ty-tup-i)}]
      {\Gamma;\Delta_1,\Delta_2\tpjg\tuple{\exp_1,\exp_2}:\ty_1*\ty_2}
      {\Gamma;\Delta_1\tpjg\exp_1:\ty_1 &
       \Gamma;\Delta_2\tpjg\exp_2:\ty_2 } \\[2pt]
\infer[\mbox{\bf(ty-fst)}]
      {\Gamma;\Delta\tpjg\dfst{\exp}:\ty_1}
      {\Gamma;\Delta\tpjg\exp:\ty_1*\ty_2}
\kern12pt
\infer[\mbox{\bf(ty-snd)}]
      {\Gamma;\Delta\tpjg\dsnd{\exp}:\ty_2}
      {\Gamma;\Delta\tpjg\exp:\ty_1*\ty_2} \\[2pt]
\infer[\mbox{\bf(ty-tup-l)}]
      {\Gamma;\Delta_1,\Delta_2\tpjg\tuple{\exp_1,\exp_2}:\vwty_1\otimes\vwty_2}
      {\Gamma;\Delta_1\tpjg\exp_1:\vwty_1 &
       \Gamma;\Delta_2\tpjg\exp_2:\vwty_2 } \\[2pt]
\infer[\mbox{\bf(ty-tup-l-elim)}]
      {\Gamma;\Delta_1,\Delta_2\tpjg\letin{\tuple{x_1,x_2}=\exp_1}{\exp_2}:\vwty}
      {$$\begin{array}{c}
       \Gamma;\Delta_1\tpjg\exp_1:\vwty_1\otimes\vwty_2 \\
       \Gamma;\Delta_2,x_1:\vwty_1,x_2:\vwty_2\tpjg\exp_2:\vwty \\
       \end{array}$$} \\[2pt]
\infer[\mbox{\bf(ty-lam-l)}]
      {\Gamma;\Delta\tpjg\lam{x}{\exp}:\vwty_1\ltimp\vwty_2}
      {(\Gamma;\Delta),x:\vwty_1\tpjg\exp:\vwty_2} \\[2pt]
\infer[\mbox{\bf(ty-app-l)}]
      {\Gamma;\Delta_1,\Delta_2\tpjg\app{\exp_1}{\exp_2}:\vwty_2}
      {\Gamma;\Delta_1\tpjg\exp_1:\vwty_1\ltimp\vwty_2 &
       \Gamma;\Delta_2\tpjg\exp_2:\vwty_1 } \\[2pt]
\infer[\mbox{\bf(ty-lam-i)}]
      {\Gamma;\emptyset\tpjg\lam{x}{\exp}:\vwty_1\itimp\vwty_2}
      {(\Gamma;\emptyset),x:\vwty_1\tpjg\exp:\vwty_2 &
       \resof(\exp) = \emptyset} \\[2pt]
\infer[\mbox{\bf(ty-app-i)}]
      {\Gamma;\Delta_1,\Delta_2\tpjg\app{\exp_1}{\exp_2}:\vwty_2}
      {\Gamma;\Delta_1\tpjg\exp_1:\vwty_1\itimp\vwty_2 &
       \Gamma;\Delta_2\tpjg\exp_2:\vwty_1 } \\[2pt]
\infer[\mbox{\bf(ty-fix)}]
      {\Gamma;\emptyset\tpjg\fix{f}{\val}:\ty}
      {\Gamma, f:\ty;\emptyset\tpjg\val:\ty} \\[2pt]
\infer[\mbox{\bf(ty-pool)}]
      {\tpjg\pool:\vwty}
      {$$\begin{array}{c}
       (\emptyset;\emptyset)\tpjg\pool(0):\vwty \\
       (\emptyset;\emptyset)\tpjg\pool(\tid):\tunit~~\mbox{for each $0<\tid\in\dom(\pool)$} \\
       \end{array}$$} \\[2pt]
\end{array}\]
\caption{The typing rules for $\langz$}
\label{figure:langz:typing_rules}
\end{figure}
\subsection{Static Semantics}
We present typing rules for $\langz$ in this section.  It is required that
each variable occur at most once in an intuitionistic (linear) expression
context $\Gamma$ ($\Delta$), and thus $\Gamma$ ($\Delta$) can be regarded
as a finite mapping.  Given $\Gamma_1$ and $\Gamma_2$ such that
$\dom(\Gamma_1)\cap\dom(\Gamma_2)=\emptyset$, we write
$(\Gamma_1,\Gamma_2)$ for the union of $\Gamma_1$ and $\Gamma_2$.  The same
notation also applies to linear expression contexts ($\Delta$).  Given an
intuitionistic expression context $\Gamma$ and a linear expression context
$\Delta$, we can form a combined expression context $(\Gamma;\Delta)$ if
$\dom(\Gamma)\cap\dom(\Delta)=\emptyset$.  Given $(\Gamma;\Delta)$, we may
write $(\Gamma;\Delta),x:\vwty$ for either $(\Gamma;\Delta,x:\vwty)$ or
$(\Gamma,x:\vwty;\Delta)$ (if $\vwty$ is actually a type).


A typing judgment in $\langz$ is of the form
$(\Gamma;\Delta)\tpjg\exp:\vwty$, meaning that $\exp$ can be assigned
the viewtype $\vwty$ under $(\Gamma;\Delta)$.  The typing rules for
$\langz$ are listed in Figure~\ref{figure:langz:typing_rules}.
In the rule $\mbox{\bf(ty-cst)}$, the following judgment requires that
the c-type be an instance of the c-type schema assigned to $\const$ in
$\SIG$:
$$\SIG\temd\const:(\vwty_1,\ldots,\vwty_n)\ctimp\vwty$$
For the constant function $\mbox{\it iadd}$ mentioned
previously, the following judgment is valid:
$$\SIG\temd\mbox{\it iadd}:(\tint(0), \tint(1))\ctimp\tint(0+1)$$
and the following judgment is valid as well:
$$\SIG\temd\mbox{\it iadd}:(\tint(2), \tint(2))\ctimp\tint(2+2)$$

By inspecting the typing rules in
Figure~\ref{figure:langz:typing_rules}, we can readily see that a
closed value cannot contain any resources if the value itself can be
assigned a type (rather than a linear type).  More formally, we have
the following proposition:
\begin%
{proposition}
\label{proposition:value_type}
If $(\emptyset;\emptyset)\tpjg\val:\ty$ is derivable,
then $\resof(\val)=\emptyset$.
\end{proposition}
This proposition plays a fundamental role in $\langz$:
The rules in Figure~\ref{figure:langz:typing_rules} are actually
so formulated in order to make it hold.

The following lemma, which is often referred to as
{\em Lemma of Canonical Forms}, relates the form of a value to its type:
\begin{lemma}\label{lemma:langz:canonical}
Assume that $(\emptyset;\emptyset)\tpjg\val:\vwty$ is derivable.
\begin{itemize}
\item
If $\vwty=\tbase$,
then $\val$ is of the form $\ccon(\val_1,\ldots,\val_n)$.
\item
If $\vwty=\vtbase$,
then $\val$ is of the form $\cres$ or $\ccon(\val_1,\ldots,\val_n)$.
\item
If $\vwty=\tunit$, then $\val$ is $\dunit$.
\item
If
$\vwty=\ty_1*\ty_2$
or
$\vwty=\vwty_1\otimes\vwty_2$,
then $\val$ is of the form $\tuple{\val_1,\val_2}$.
\item
If $\vwty=\vwty_1\itimp\vwty_2$ or $\vwty=\vwty_1\ltimp\vwty_2$, then
$\val$ is of the form $\lam{x}{\exp}$.
\end{itemize}
\end{lemma}
\begin{proof}
By an inspection of the rules in Figure~\ref{figure:langz:typing_rules}.
\end{proof}

We use $\esubst$ for substitution on variables $\xf$:
\[\begin{array}{rcl}
\esubst & ::= & [] \mid \esubst[x\mapsto\val] \mid \esubst[f\mapsto e] \\
\end{array}\]
For each $\esubst$, we define the multiset $\resof(\esubst)$ of resources
in $\esubst$ as follows:
$$\resof(\esubst)=\uplus_{\xf\in\dom(\esubst)}\resof(\esubst(\xf))$$
Given an expression $\exp$, we use $\exp[\esubst]$ for the result of
applying $\esubst$ to $\exp$, which is defined in a standard manner.
We write $(\Gamma_1;\Delta_1)\tpjg\esubst:(\Gamma_2;\Delta_2)$ to
mean that
\begin{itemize}
\item
$\dom(\esubst)=\dom(\Gamma_2)\cup\dom(\Delta_2)$, and
\item
$(\Gamma_1;\emptyset)\tpjg\esubst(\xf):\Gamma_2(\xf)$
is derivable for each $\xf\in\Gamma_2$, and
\item
there exists a linear expression context $\Delta_{1,x}$ for each
$x\in\dom(\Delta_2)$ such that
$(\Gamma_1;\Delta_{1,x})\tpjg\esubst(x):\Delta_2(x)$
is derivable, and
\item
$\Delta_1=\cup_{x\in\dom(\Delta_2)}\Delta_{1,x}$
\end{itemize}
The following lemma, which is often referred to as
{\em Substitution Lemma}, is needed to establish the soundness of the type
system of $\langz$:
\begin%
{lemma}\label{lemma:langz:substitution}
Assume $(\Gamma_1;\Delta_1)\tpjg\esubst:(\Gamma_2;\Delta_2)$ and
$(\Gamma_2;\Delta_2)\tpjg\exp:\vwty$.
Then $(\Gamma_1;\Delta_1)\tpjg\exp[\esubst]:\vwty$ is derivable
and $\resof(\exp[\esubst])=\resof(\exp)\uplus\resof(\esubst)$.
\end{lemma}
\begin{proof}
By induction on the derivation of $(\Gamma_2;\Delta_2)\tpjg\exp:\vwty$.
\end{proof}

\subsection
{Dynamic Semantics}
\label{section:langz:dynamics}
We present evaluation rules for $\langz$ in this section.
The evaluation contexts in $\langz$ are defined below:
\[\begin{array}{lrcl}
\mbox{eval.~ctx.}
& E & ::= & \kern144pt \\
\kern12pt
\hbox to 0pt
{$[] \mid \const(\vval,E,\vexp) \mid \dif(E,\exp_1,\exp_2) \mid$ \hss} \\
\kern12pt
\hbox to 0pt
{$\tuple{E,\exp} \mid \tuple{\val, E} \mid \letin{\tuple{x_1,x_2}=E}{\exp} \mid$} \\
\kern12pt
\hbox to 0pt
{$\dfst{E} \mid \dsnd{E} \mid \app{E}{\exp} \mid \app{\val}{E}$} \\
\end{array}\]
Given an evaluation context $E$ and an expression $\exp$, we use $E[\exp]$
for the expression obtained from replacing the only hole $[]$ in $E$ with
$\exp$.

\def\subst#1#2#3{#3[#2\mapsto #1]}
\begin%
{definition}
We define pure redexes and their reducts as follows.
\begin%
{itemize}
\item
$\dif(\ctrue,\exp_1,\exp_2)$ is a pure redex whose reduct is $\exp_1$.
\item
$\dif(\cfalse,\exp_1,\exp_2)$ is a pure redex whose reduct is $\exp_2$.
\item
$\letin{\tuple{x_1,x_2}=\tuple{\val_1,\val_2}}{\exp}$
is a pure redex whose reduct is $\subst{\val_1,\val_2}{x_1,x_2}{\exp}$.
\item
$\dfst{\tuple{\val_1,\val_2}}$ is a pure redex whose reduct is $\val_1$.
\item
$\dsnd{\tuple{\val_1,\val_2}}$ is a pure redex whose reduct is $\val_2$.
\item
$\app{\lam{x}{\exp}}{\val}$
is a pure redex whose reduct is $\subst{\val}{x}{\exp}$.
\item
$\fix{f}{\val}$ is a pure redex whose reduct is $\subst{\fix{f}{\val}}{f}{\val}$.
\end{itemize}
\end{definition}

\def\eval{\rightarrow}
\def\Eval{\Rightarrow}
\def\frandbit{\mbox{\it randbit}}
Evaluating calls to constant functions is of particular importance in
$\langz$. Assume that $\cfun$ is a constant function of arity $n$. The
expression $\cfun(v_1,\ldots,v_n)$ is an {\em ad-hoc} redex if $\cfun$ is
defined at $v_1,\ldots,v_n$, and any value of $\cfun(v_1,\ldots,v_n)$ is a
reduct of $\cfun(v_1,\ldots,v_n)$. For instance, $1+1$ is an ad hoc redex
and $2$ is its sole reduct. In contrast, $1+\ctrue$ is not a redex as it is
undefined. We can even have non-deterministic constant functions.  For
instance, we may assume that the ad-hoc redex $\frandbit()$ can evaluate to
both 0 and 1.

Let $\exp$ be a well-typed expression of the form
$\cfun(v_1,\ldots,v_n)$ and $\resof(\exp)\subseteq\rns$ holds for some
valid $\rns$ (that is, $\rns\in\RES$).  We always assume that there
exists a reduct $\val$ in $\langz$ for $\cfun(v_1,\ldots,v_n)$ such
that $(\rns\backslash\resof(e))\uplus\resof(v)\in\RES$. By doing so,
we are able to give a presentation with much less clutter.

\begin%
{definition} Given expressions $\exp_1$ and $\exp_2$, we write
$\exp_1\eval\exp_2$ if $\exp_1=E[\exp]$ and $\exp_2=E[\exp']$ for some
$E,\exp$ and $\exp'$ such that $\exp'$ is a reduct of $\exp$, and we may
say that $\exp_1$ evaluates or reduces to $\exp_2$ purely if $\exp$ is a
pure redex.
\end{definition}

Note that resources may be generated as well as consumed when ad-hoc
reductions occur. This is an essential issue of great importance in any
linear type system designed to support practical programming.

\begin%
{definition}
Given pools $\pool_1$ and $\pool_2$, the relation $\pool_1\eval\pool_2$ is
defined according to the following rules:\\[-18pt]
\begin{center}
\fontsize{8}{8}\selectfont
\[\begin{array}{c}
\infer[\mbox{(PR0)}]
      {\pool[\tid\mapsto\exp_1]\eval\pool[\tid\mapsto\exp_2]}
      {\exp_1\eval\exp_2} \\[2pt]
\infer[\mbox{(PR1)}]
      {\pool\eval\pool[\tid_0:=E[\tuple{}]][\tid\mapsto\app{\lam{x}{e}}{\tuple{}}]}
      {\pool(\tid_0)=E[\fthreadcreate(\lam{x}{e})]} \\[2pt]
\infer[\mbox{(PR2)}]
      {\pool[\tid\mapsto\tuple{}]\eval\pool}{\tid > 0}
\end{array}\]
\end{center}
\end{definition}
If a pool $\pool_1$ evaluates to another pool $\pool_2$ by the rule (PR0),
then one program in $\pool_1$ evaluates to its counterpart in $\pool_2$ and
the rest stay the same; if by the rule (PR1), then a fresh program is
created; if by the rule (PR2), then a program (that is not the main
program) is eliminated.

From this point on, we always (implicitly) assume that
$\resof(\pool)\in\RES$ holds whenever $\pool$ is well-typed.  The soundness
of the type system of $\langz$ rests upon the following two theorems:
\begin%
{theorem}
(Subject Reduction on Pools)
\label{theorem:langz:subject_reduction_on_pools}
Assume that $\tpjg\pool_1:\vwty$ is derivable and $\pool_1\eval\pool_2$
holds for some $\pool_2$ satisfying $\resof(\pool_2)\in\RES$. Then
$\tpjg\pool_2:\vwty$ is also derivable.
\end{theorem}
\begin%
{proof}
By structural induction on the derivation of $\tpjg\pool_1:\vwty$.
Note that Lemma~\ref{lemma:langz:substitution} is needed.
\end{proof}

\begin%
{theorem}
(Progress Property on Pools)
\label{theorem:langz:progress_on_pools}
Assume that $\tpjg\pool_1:\vwty$ is derivable. Then we have the following
possibilities:
\begin{itemize}
\item
$\pool_1$ is a singleton mapping $[0\mapsto\val]$ for some $\val$, or
\item
$\pool_1\eval\pool_2$ holds for some $\pool_2$ such that $\resof(\pool_2)\in\RES$.
\end{itemize}
\end{theorem}
\begin%
{proof}
By structural induction on the derivation of $\tpjg\pool_1:\vwty$.  Note
that Lemma~\ref{lemma:langz:canonical} is needed. Essentially, we can
readily show that $\Pi_1(\tid)$ for any $\tid\in\dom(\Pi_1)$ is either a
value or of the form $E[\exp]$ for some evaluation context $E$ and redex
$\exp$.  If $\Pi_1(\tid)$ is a value for some $\tid>0$, then this value
must be $\tuple{}$.  So the rule $\mbox{(PR2)}$ can be used to reduce
$\Pi_1$.  If $\Pi_1(\tid)$ is of the form $E[e]$ for some redex $e$, then
the rule $\mbox{(PR0)}$ can be used to reduce $\Pi_1$.
\end{proof}

By combining Theorem~\ref{theorem:langz:subject_reduction_on_pools} and
Theorem~\ref{theorem:langz:progress_on_pools}, we immediately conclude that
the evaluation of a well-typed pool either leads to a pool that itself is a
singleton mapping of the form $[0\mapsto\val]$ for some value $\val$, or it
goes on forever.  In other words, $\langz$ is type-sound.

\section%
{Extending $\langz$ with Channels}
\label{section:langch}
There is no support for communication between threads in $\langz$, making
$\langz$ uninteresting as a multi-threaded language. We extend $\langz$ to
$\langch$ with support for synchronous communication channels in this
section. Supporting asynchronous communication channels is certainly
possible but would result in a more involved theoretical development.
We do support both synchronous and asynchronous session-typed communication
channels in practice, though. In order to assign types to channels, we
introduce session types as follows:
\[
\begin%
{array}{rcl}
\ST & ::= & \stnil \mid \chmsg{i,j}{\ST} \\
\end{array}
\]
An empty session is specified by $\stnil$.  Given integers $i$ and $j$
(representing roles), the precise meaning of the term $\stmsg(i,j)$ is
to be given later, which depends on the group of the roles implemented
by a party.  Intuitively speaking, this term refers to transferring a
message of some kind from a party implementing the role $i$ (and
possibly others) to another one implementing the role $j$ (and
possibly others).  Please notice that $\stmsg(i,j)$ is written instead
of $\stmsg(i,j,\vwty)$ for some viewtype $\vwty$. The omission of
$\vwty$ is solely for the purpose of a simplified presentation as the
primary focus is on communications between parties in a session
(rather than values transferred during communications).

\def\nrole{\mbox{\it nrole}}
Let us assume the availability of finite sets of integers for forming
types. As another step towards a simplified presentation, we fix a set
of roles $0,1,\ldots,\nrole-1$ for some natural number $\nrole\geq 2$
and require that the roles mentioned in every session type belong to
this set. Given a group $G$ of roles and a session type $\ST$, we can
form a linear base viewtype $\tchan(G,\ST)$ for channels that are
often referred to as $G$-channels; the party in a session that holds
a $G$-channel is supposed to implement all of the roles in $G$.

The function $\fchancreate$ for creating a channel is assigned the
following c-type schema:
\[
\begin%
{array}{rcl}
\fchancreate & : &
(\tchan(G,\ST)\ltimp\tunit)\ctimp\tchan(\setcomp{G},\ST) \\
\end{array}\]
where $G\neq\emptyset$ and $\setcomp{G}\neq\emptyset$ is assumed.
Given a linear function of the type $\tchan(G,\ST)\ltimp\tunit$ for
some session type $\ST$, $\fchancreate$ essentially creates two
properly connected channels of the types $\tchan(G,\ST)$ and
$\tchan(\setcomp{G},\ST)$, and then starts a thread for evaluating
the call that applies the function to the channel of the type
$\tchan(G,\ST)$ and then returns the other channel of the type
$\tchan(\setcomp{G},\ST)$.  The newly created two channels share
the same id.

The function for sending onto a channel is given the following type
schema:
$$\fsend~~:~~(\tchan(G,\chmsg{i,j}{\ST}))\ctimp\tchan(G,\ST)$$
where $i\in G$ and $j\not\in G$ is assumed. The function for
receiving from a channel is given the following type schema:
$$\frecv~~:~~(\tchan(G,\chmsg{i,j}{\ST}))\ctimp\tchan(G,\ST)$$
where $i\not\in G$ and $j\in G$ is assumed.
The function for skipping an internal or external message
is given the following type schema:
$$\fskip~~:~~(\tchan(G,\chmsg{i,j}{\ST}))\ctimp\tchan(G,\ST)$$
where either
$i\in G$ and $j\in G$ is assumed or $i\not\in G$ and $j\not\in G$ is
assumed.  The function for closing a channel is given the following
type schema:
$$\fclose~~:~~(\tchan(G,\ST))\ctimp\tunit$$
Note that $\fsend$ and $\frecv$ correspond to the functions
$\mbox{\tt channel\_send}$ and $\mbox{\tt channel\_recv}$ mentioned in
Section~\ref{section:introduction}, respectively, and $\fclose$
corresponds to $\mbox{\tt channel\_close}$.

\def\chpres{\mbox{ch}^{+}}
\def\chnres{\mbox{ch}^{-}}
In $\langch$, there are resource constants
$\chpres_n$ and $\chnres_n$ referring to positive and negative
channels, respectively, where $n$ ranges over natural numbers.  For
each $n$, $\chpres_n$ and $\chnres_n$ are dual to each other and their
channel ids are $n$.  We use $\chpcst$ and $\chncst$ for positive and
negative channels, and $\chcst$ for both.  If $\chpcst$ and $\chncst$
appear in the same context, it is assumed (unless specified otherwise)
that they refer to $\chpres_n$ and $\chnres_n$ for the same id $n$.

Given a group $G$ of roles, a $G$-channel is positive if $0\in G$
and it is negative if $0\not\in G$. Note that calling $\fchancreate$
creates a positive channel and a negative channel of the same id; one
is passed to a newly created thread while the other is returned to the
caller.

There are no new typing rules in $\langch$ over $\langz$.
Given $G$ and $\ST$, we say that the type $\tchan(G,\ST)$
matches the type $\tchan(\setcomp{G},\ST)$ and vice versa. In any
type derivation of $\pool:\vwty$ satisfying $\resof(\pool)\in\RES$,
the type assigned to a positive channel $\chpcst$ is always required
to match the one assigned to the corresponding negative channel
$\chncst$ of the same channel id.
For evaluating pools in $\langz$, we have the following additional rules in
$\langch$:\\[-18pt]
\begin{center}
\fontsize{8}{8}\selectfont
\[\begin{array}{c}
\infer[\mbox{(PR3)}]
      {\pool\eval\pool[\tid_0:=E[\chncst]][\tid\mapsto\app{\lam{x}{e}}{\chpcst}]}
      {\pool(\tid_0)=E[\fchancreate(\lam{x}{e})]} \\[2pt]
\infer[\mbox{(PR4-send)}]
      {\pool\eval\pool[\tid_1:=E_1[\chpcst]][\tid_2:=E_2[\chncst]]}
      {\pool(\tid_1)=E_1[\fsend(\chpcst)]&&\pool(\tid_2)=E_2[\frecv(\chncst)]} \\[2pt]
\infer[\mbox{(PR4-recv)}]
      {\pool\eval\pool[\tid_1:=E_1[\chpcst]][\tid_2:=E_2[\chncst]]}
      {\pool(\tid_1)=E_1[\frecv(\chpcst)]&&\pool(\tid_2)=E_2[\fsend(\chncst)]} \\[2pt]
\infer[\mbox{(PR4-skip)}]
      {\pool\eval\pool[\tid_1:=E_1[\chpcst]][\tid_2:=E_2[\chncst]]}
      {\pool(\tid_1)=E_1[\fskip(\chpcst)]&&\pool(\tid_2)=E_2[\fskip(\chncst)]} \\[2pt]
\infer[\mbox{(PR4-close)}]
      {\pool\eval\pool[\tid_1:=E_1[\dunit]][\tid_2:=E_2[\dunit]]}
      {\pool(\tid_1)=E_1[\fclose(\chpcst)]&&\pool(\tid_2)=E_2[\fclose(\chncst)]} \\[2pt]
\end{array}\]
\end{center}
For instance, the rule $\mbox{PR4-send}$ states:
If a program in a pool is of the form $E_1[\fsend(\chpcst)]$ and
another of the form $E_2[\frecv(\chncst)]$, then this pool can be
reduced to another pool by replacing these two programs with
$E_1[\chpcst]$ and $E_2[\chncst]$, respectively.

While Theorem~\ref{theorem:langz:subject_reduction_on_pools} (Subject
Reduction) can be readily established for $\langch$,
Theorem~\ref{theorem:langz:progress_on_pools} (Progress) requires some
special treatment due to the presence of session-typed primitive functions
$\fchancreate$, $\fsend$, $\frecv$, $\fskip$, and $\fclose$.

A partial (ad-hoc) redex in $\langch$ is of one of the following
forms: $\fsend(\chcst)$, $\frecv(\chcst)$, $\fskip(\chcst)$, and
$\fclose(\chcst)$.  We say that $\fsend(\chpcst)$ and
$\frecv(\chncst)$ match, and $\frecv(\chpcst)$ and $\fsend(\chncst)$
match, and $\fskip(\chpcst)$ and $\fskip(\chncst)$ match, and
$\fclose(\chpcst)$ and $\fclose(\chncst)$ match.  We can immediately
prove in $\langch$ that each well-typed program is either a value or
of the form $E[\exp]$ for some evaluation context $E$ and expression
$\exp$ that is either a redex or a partial redex. We refer to an
expression as a {\it blocked} one if it is of the form $E[\exp]$ for
some partial redex $\exp$. We say two blocked expressions
$E_1[\exp_1]$ and $E_2[\exp_2]$ match if $\exp_1$ and $\exp_2$ are
matching partial redexes. Clearly, a pool containing two matching
blocked expressions can be reduced according to one of the rules
$\mbox{PR4-send}$, $\mbox{PR4-recv}$, $\mbox{PR4-skip}$, and
$\mbox{PR4-close}$.

Intuitively, a pool $\Pi$ is deadlocked if $\Pi(tid)$ for
$\tid\in\dom(\Pi)$ are all blocked expressions but there are no matching
ones among them, or if $\Pi(0)$ is a value and $\Pi(tid)$ for positive
$\tid\in\dom(\Pi)$ are all blocked expressions but there are no matching
ones among them. The following lemma states that a well-typed pool in
$\langch$ can never be deadlocked:

\begin%
{lemma}
(Deadlock-Freedom)
\label{lemma:langch:deadlock-freedom}
Let $\Pi$ be a well-typed pool in $\langch$ such that $\Pi(0)$ is either a
value containing no channels or a blocked expression and $\Pi(tid)$ for
each positive $\tid\in\dom(\Pi)$ is a blocked expression. If $\Pi$ is
obtained from evaluating an initial pool containing no channels, then there
exist two thread ids $\tid_1$ and $\tid_2$ such that $\Pi(\tid_1)$ and
$\Pi(\tid_2)$ are matching blocked expressions.
\end{lemma}
Note that it is entirely possible to encounter a scenario where the main
thread in a pool returns a value containing a channel while another thread
is waiting for something to be sent on the channel. Technically, we do not
classify this scenario as a deadlocked one. There are many forms of values
that contain channels. For instance, such a value can be a channel itself,
or a closure-function containing a channel in its environment, or a compound
value like a tuple that contains a channel as one part of it, etc. Clearly,
any value containing a channel can only be assigned a true viewtype.

\def\MCH{M}
\def\MCHS{{\cal M}}
\def\DFred{\leadsto}
As a channel can be sent from one thread to another one,
establishing Lemma~\ref{lemma:langch:deadlock-freedom} is
conceptually challenging. The following technical approach to
addressing the challenge is adopted from some existing work on
dyadic sessions types (for i-sessions)~\cite{MTLC-i-sessiontype}.
One may want skip the rest of this section when reading the
paper for the first time.

Let us use $\MCH$ for sets of (positive and negative) channels and $\MCHS$
for a finite non-empty collection (that is, multiset) of such sets. We say
that $\MCHS$ is {\it regular} if the sets in $\MCHS$ are pairwise disjoint
and each pair of channels $\chpcst$ and $\chncst$ are either both included
in the multiset union $\biguplus(\MCHS)$ of all the sets in $\MCHS$ or both
excluded from it. Of course, $\biguplus(\MCHS)$ is the same as the set
union $\bigcup(\MCHS)$ as the sets in $\MCHS$ are pairwise disjoint.

Let $\MCHS$ be a regular collection of channel sets. We say that $\MCHS$
{\it DF-reduces} to $\MCHS'$ via $\chpcst$ if there exist $\MCH_1$ and
$\MCH_2$ in $\MCHS$ such that $\chpcst\in\MCH_1$ and $\chncst\in\MCH_2$ and
$\MCHS'=(\MCHS\backslash\{\MCH_1,\MCH_2\})\cup\{\MCH_{12}\}$, where
$\MCH_{12}=(\MCH_1\cup\MCH_2)\backslash\{\chpcst,\chncst\}$.  We say that
$\MCHS$ DF-reduces to $\MCHS'$ if $\MCHS$ DF-reduces to $\MCHS'$ via some
$\chpcst$. We may write $\MCHS\DFred\MCHS'$ to mean that $\MCHS$
DF-reduces to $\MCHS$. We say that $\MCHS$ is DF-normal if there is no $\MCHS'$
such that $\MCHS\DFred\MCHS'$ holds.

\begin%
{proposition}
Let $\MCHS$ be a regular collection of channel sets.  If $\MCHS$ is
DF-normal, then each set in $\MCHS$ consists of an indefinite number of
channel pairs $\chpcst$ and $\chncst$. In other words, for each $\MCH$
in a DF-normal $\MCHS$, a channel $\chpcst$ is in $\MCH$ if and only if
its dual $\chncst$ is also in $\MCH$.
\end{proposition}
\begin
{proof}
The proposition immediately follows from the definition of DF-reduction
$\leadsto$.
\end{proof}

\begin%
{definition}
\label{def:DFreducibility}
A regular collection $\MCHS$ of channel sets is DF-reducible if either (1)
each set in $\MCHS$ is empty or (2) $\MCHS$ is not DF-normal and $\MCHS'$
is DF-reducible whenever $\MCHS\DFred\MCHS'$ holds.
\end{definition}
We say that a channel set $\MCH$ is self-looping if it contains both
$\chpcst$ and $\chncst$ for some $\chpcst$.  Obviously, a regular collection
$\MCHS$ of channel sets is not DF-reducible if there is a self-looping $\MCH$
in $\MCHS$.

\begin%
{proposition}
\label{prop:DFreducibility:0}
Let $\MCHS$ be a regular collection of channel sets.  If
$\MCHS$ is DF-reducible and $\MCHS'=\MCHS\backslash\{\emptyset\}$,
then $\MCHS'$ is also DF-reducible.
\end{proposition}
\begin
{proof}
Straightforwardly.
\end{proof}

\begin%
{proposition}
\label{prop:DFreducibility:1}
Let $\MCHS$ be a regular collection of channel sets.  If
$\MCHS\DFred\MCHS'$ and $\MCHS'$ is DF-reducible, then $\MCHS$ is also
DF-reducible.
\end{proposition}
\begin
{proof}
Clearly, $\MCHS\DFred\MCHS'$ via some $\chpcst$.  Assume
$\MCHS\DFred\MCHS_1$ via $\chpcst_1$ for some $\MCHS_1$ and $\chpcst_1$.
If $\chpcst$ and $\chpcst_1$ are the same, then $\MCHS_1$ is DF-reducible
as it is the same as $\MCHS'$. Otherwise, it can be readily verified that
there exists $\MCHS'_1$ such that $\MCHS_1\DFred\MCHS'_1$ via $\chpcst$ and
$\MCHS'\DFred\MCHS'_1$ via $\chpcst_1$. Clearly, the latter implies
$\MCHS'_1$ being DF-reducible.  Note that the size of $\MCHS_1$ is strictly
less than that of $\MCHS$.  By induction hypothesis on $\MCHS_1$, we have
$\MCHS_1$ being DF-reducible.  By definition, $\MCHS$ is DF-reducible.
\end{proof}

\begin%
{proposition}
\label{prop:DFreducibility:2}
Let $\MCHS$ be a regular collection of channel sets that is DF-reducible.
If $\MCH_1$ and $\MCH_2$ in $\MCHS$ contain $\chpcst$ and $\chncst$,
respectively, then
$\MCHS'=(\MCHS\backslash\{\MCH_1,\MCH_2\})\cup\{\MCH'_1,\MCH'_2\}$ is also
DF-reducible, where $\MCH'_1=\MCH_1\backslash\{\chpcst\}$ and
$\MCH'_2=\MCH_2\backslash\{\chncst\}$.
\end{proposition}
\begin%
{proof}
The proposition follows from a straightforward induction on the size of
the set union $\bigcup(\MCHS)$.
\end{proof}
\begin%
{lemma}
\label{lemma:DFreducibility:3}
Let $\MCHS$ be a regular collection of $n$ channel sets
$\MCH_1,\ldots,\MCH_n$ for some $n\geq 1$.  If the union
$\bigcup(\MCHS)=\MCH_1\cup\ldots\cup\MCH_n$ contains at least $n$ channel
pairs $(\chpcst_1,\chncst_1),\ldots,(\chpcst_n,\chncst_n)$, then $\MCHS$ is
not DF-reducible.
\end{lemma}
\begin%
{proof}
By induction on $n$.  If $n=1$, then $\MCHS$ is not DF-reducible as
$\MCH_1$ is self-looping.  Assume $n > 1$.  If either $\MCH_1$ or $\MCH_2$
is self-looping, then $\MCHS$ is not DF-reducible.  Otherwise, we may
assume that $\chpcst_1\in\MCH_1$ and $\chncst_1\in\MCH_2$ without loss of
generality.  Then $\MCHS$ DF-reduces to $\MCHS'$ via $\chpcst_1$ for some
$\MCHS'$ containing $n-1$ channel sets. Note that $\bigcup(\MCHS')$
contains at least $n-1$ channel pairs
$(\chpcst_2,\chncst_2),\ldots,(\chpcst_n,\chncst_n)$. By induction
hypothesis, $\MCHS'$ is not DF-reducible. So $\MCHS$ is not DF-reducible,
either.
\end{proof}
\def\fMCH{\rho_{\it CH}}
\def\fMCHS{{\cal R}_{\it CH}}
Given an expression $\exp$ in $\langch$, we use
$\fMCH(\exp)$ for the set of channels contained in $\exp$.
Given a pool $\Pi$ in $\langch$, we use $\fMCHS(\Pi)$ for the
collection of $\fMCH(\Pi(tid))$, where $\tid$ ranges over $\dom(\Pi)$.
\begin%
{lemma}
\label{lemma:DFreducibility:4}
If $\fMCHS(\Pi)$ is DF-reducible and $\Pi$ evaluates to $\Pi'$, then
$\fMCHS(\Pi')$ is also DF-reducible.
\end{lemma}
\begin%
{proof}
Note that $\fMCHS(\Pi)$ and $\fMCHS(\Pi')$ are the same unless $\Pi$
evaluates to $\Pi'$ according to one of the rules $\mbox{PR3}$,
$\mbox{PR4-send}$, $\mbox{PR4-recv}$, $\mbox{PR4-skip}$, and
$\mbox{PR4-close}$.
\begin%
{itemize}
\item
For the rule $\mbox{PR3}$: We have $\fMCHS(\Pi')\DFred\fMCHS(\Pi)$ via the
newly introduced channel $\chpcst$. By
Proposition~\ref{prop:DFreducibility:1}, $\fMCHS(\Pi')$ is DF-reducible.
\item
For the rule $\mbox{PR4-send}$: Let $\chpcst$ be the channel on which a
value is sent when $\Pi$ evaluates to $\Pi'$. Note that this value can
itself be a channel or contain a channel. We have $\fMCHS(\Pi)\DFred\MCHS$
via $\chpcst$ for some $\MCHS$. So $\MCHS$ is DF-reducible by definition.
Clearly, $\fMCHS(\Pi')\DFred\MCHS$ via $\chpcst$ as well.  By
Proposition~\ref{prop:DFreducibility:2}, $\fMCHS(\Pi')$ is DF-reducible.
\item
For the rule $\mbox{PR4-recv}$: This case is similar to the previous one.
\item
For the rule $\mbox{PR4-skip}$: This case is trivial as $\fMCHS(\Pi)$ and
$\fMCHS(\Pi')$ are the same.
\item
For the rule $\mbox{PR4-close}$: We have that $\fMCHS(\Pi')$ is DF-reducible
by Proposition~\ref{prop:DFreducibility:2}.
\end{itemize}
In order to fully appreciate the argument made in the case of
$\mbox{PR4-send}$, one needs to imagine a scenario where a channel is
actually transferred from one thread into another. While this scenario
does not happen here due to the simplified version of type schemas
assigned to $\fsend$ and $\frecv$, one can find the essential details
in the original paper on DF-reducibility~\cite{MTLC-i-sessiontype}.
\end{proof}
We are now ready to give a proof for
Lemma~\ref{lemma:langch:deadlock-freedom}:
\vspace{6pt}

\noindent%
{\it Proof}~~
Note that any channel, either positive or negative, can appear at most
once in $\fMCHS(\Pi)$, and a channel $\chpcst$ appears in
$\fMCHS(\Pi)$ if and only if its dual $\chncst$ also appears in
$\fMCHS(\Pi)$. In addition, any positive channel $\chpcst$ being
assigned a type of the form $\tchan(G,\ST)$ in the type derivation of
$\Pi$ for some session type $S$ mandates that its dual $\chncst$ be
assigned the type of the form $\tchan(\setcomp{G},\ST)$.

Assume that $\Pi(\tid)$ is a blocked expression for each
$\tid\in\dom(\Pi)$. If the partial redex in $\Pi(\tid_1)$ involves a
positive channel $\chpcst$ while the partial redex in $\Pi(\tid_2)$
involves its dual $\chncst$, then these two partial redexes must
match. This is due to $\Pi$ being well-typed. In other words, the ids
of the channels involved in the partial redexes of $\Pi(\tid)$ for
$\tid\in\dom(\Pi)$ are all distinct. This simply implies that there
are $n$ channel pairs $(\chpcst,\chncst)$ in $\bigcup(\fMCHS(\Pi))$
for some $n$ greater than or equal to the size of $\Pi$.  By
Lemma~\ref{lemma:DFreducibility:3}, $\fMCHS(\Pi)$ is not reducible.
On the other hand, $\fMCHS(\Pi)$ is reducible by
Lemma~\ref{lemma:DFreducibility:4} as $\Pi_0$ evaluates to $\Pi$ (in
many steps) and $\fMCHS(\Pi_0)$ (containing only sets that are empty)
is reducible. This contradiction indicates that there exist $\tid_1$
and $\tid_2$ such that $\Pi(\tid_1)$ and $\Pi(\tid_2)$ are matching
blocked expressions. Therefore $\Pi$ evaluates to $\Pi'$ for some pool
$\Pi'$ according to one of the rules $\mbox{PR4-clos}$,
$\mbox{PR4-send}$, and $\mbox{PR4-recv}$.

With Proposition~\ref{prop:DFreducibility:0}, the case can be handled
similarly where $\Pi(0)$ is a value containing no channels and $\Pi(\tid)$
is a blocked expression for each positive
$\tid\in\dom(\Pi)$.~\hfill\fillsquare

\vspace{6pt}
Please assume for the moment that we would like to add into $\langch$ a function
$\fchancreatetwo$ of the following type schema:
\[
\fontsize{8pt}{9pt}\selectfont
\begin%
{array}{c}
((\tchan(G_1,\ST_1),\tchan(G_2,\ST_2))\ltimp\tunit)\ctimp(\tchan(\setcomp{G_1},\ST_1),\tchan(\setcomp{G_2},\ST_2)) \\
\end{array}\]
One may think of $\fchancreatetwo$ as a ``reasonable'' generalization of
$\fchancreate$ that creates in a single call two channels instead of one.
Unfortunately, adding $\fchancreatetwo$ into $\langch$ can potentially
cause a deadlock. For instance, we can easily imagine a scenario where the
first of the two channels $(\chncst_1,\chncst_2)$ returned from a call to
$\fchancreatetwo$ is used to send the second to the newly created thread
by the call, making it possible for that thread to cause a deadlock by
waiting for a value to be sent on $\chpcst_2$.  Clearly,
Lemma~\ref{lemma:DFreducibility:4} is invalidated if $\fchancreatetwo$ is
added.

\vspace{6pt}
\noindent%
The soundness of the type system of $\langch$ rests upon the following
two theorems
(corresponding to
Theorem~\ref{theorem:langz:subject_reduction_on_pools}
and Theorem~\ref{theorem:langz:progress_on_pools}):
\begin%
{theorem}
(Subject Reduction on Pools)
\label{theorem:langch:subject_reduction_on_pools}
Assume that $\tpjg\pool_1:\vwty$ is derivable and
$\pool_1\eval\pool_2$ such that $\resof(\pool_2)\in\RES$.
Then $\tpjg\pool_2:\vwty$ is derivable.
\end{theorem}
\begin%
{proof}
The proof is essentially the same as the one for
Theorem~\ref{theorem:langz:subject_reduction_on_pools}.  The only
additional part is for checking that the rules $\mbox{PR3}$,
$\mbox{PR4-clos}$, $\mbox{PR4-send}$, and $\mbox{PR4-recv}$ are all
consistent with respect to the typing rules listed in
Figure~\ref{figure:langz:typing_rules}.
\end{proof}

\begin%
{theorem}
(Progress Property on Pools)
\label{theorem:langch:progress_on_pools}
Assume that $\tpjg\pool_1:\vwty$ is derivable and
$\resof(\pool_1)$ is valid.
Also assume that $\resof(v)$ contains no channels
for every value $v$ of the type $\vwty$. Then we have
the following possibilities:
\begin{itemize}
\item
$\pool_1$ is a singleton mapping $[0\mapsto\val]$ for some $\val$, or
\item
$\pool_1\eval\pool_2$ holds for some $\pool_2$ such that $\resof(\pool_2)\in\RES$.
\end{itemize}
\end{theorem}
\begin%
{proof}
The proof follows the same structure as the one for
Theorem~\ref{theorem:langz:progress_on_pools}.
Lemma~\ref{lemma:langch:deadlock-freedom} is needed to handle the
case where all of the threads (possibly excluding the main thread) in a
pool consist of blocked expressions.
\end{proof}

\section%
{From Dyadic to Multiparty}
\label{section:dyadic_to_multiparty}
In this section, we present an approach to building multiparty
sessions based on dyadic g-sessions. With this approach, we give
justification in support of the theoretical development in
Section~\ref{section:langz} and Section~\ref{section:langch}.

\begin{figure}[htp]
\centering{%
\includegraphics[scale=0.5]{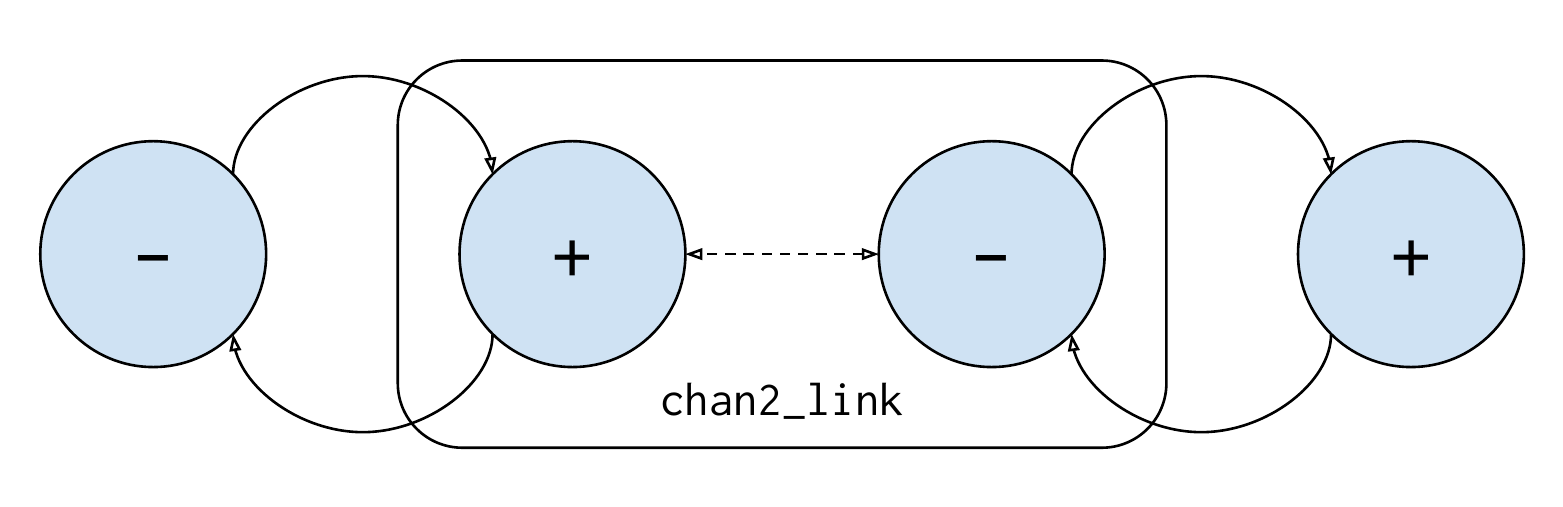}%
}
\caption{Illustrating $\fchanlinktwo$}
\label{figure:illustration_for_chan2_link}
\end{figure}
\subsection%
{Bidirectional Forwarding between Two Parties}
Given two channels of dual types, there is a generic method for
forwarding onto one channel each message received from the other
channel and vice versa. In formalizations of session types that are
directly based on linear logic(e.g.,~\cite{CairesP10,Wadler12}), this
form of bidirectional forwarding of messages corresponds to the
cut-elimination process in linear logic%
\footnote{%
It should be noted in this context that it is of no concern as to
whether the cut-elimination process is terminating or not. Instead,
the focus is solely on proving that each cut involving a compound
formula can be reduced into one or more cuts (so as to make progress).}.

\begin%
{theorem}
\label{theorem:chan2_link}
Assume that $\chcst_0$ and $\chcst_1$ are two channels of the types
$\tchan(G,\ST)$ and $\tchan(\setcomp{G},\ST)$, respectively.  For any
party holding $\chcst_0$ and $\chcst_1$, there is a generic method for
forwarding onto $\chcst_0$ the messages received from $\chcst_1$ and
vice versa (during the evaluation of a well-typed program). Let us use
the name $\fchanlinktwo$ for a function of the following type that
implements this generic method:
$$
(\tchan(G,\ST),\tchan(\setcomp{G},\ST))\ctimp\tunit
$$
\end{theorem}
\begin%
{proof}
If $\ST$ is $\stnil$, then all that is needed is to
call $\fclose$ on both $\chcst_0$ and $\chcst_1$.
Assume that
$\ST$ is of the form
$\chmsg{i,j}{\ST_1}$. Then we have the following 4 possibilities.
\begin%
{itemize}
\item
Assume $i\in G$ and $j\in G$.  Then $\stmsg(i,j)$ indicates an
internal message for $\chcst_0$ and an external message for
$\chcst_1$. So this case is handled by calling $\fskip$ on $\chcst_0$
and $\chcst_1$.
\item
Assume $i\in G$ and $j\in\setcomp{G}$.
Then $\stmsg(i,j)$ indicates sending for $\chcst_0$
and receiving for $\chcst_1$. So this case is handled by
calling $\frecv$ on $\chcst_1$ to receive a message and then
calling $\fsend$ on $\chcst_0$ to send the message.
\item
Assume $i\in\setcomp{G}$ and $j\in G$.
This case is similar to the one where $i\in G$ and $j\in\setcomp{G}$.
\item
Assume $i\in\setcomp{G}$ and $j\in\setcomp{G}$.
This case is similar to the one where $i\in G$ and $j\in G$.
\end{itemize}
After one of the above 4 possibilities is performed, a recursive
call can be made on $\chcst_0$ and $\chcst_1$ to perform the rest of
bidirectional forwarding.
\end{proof}
An illustration of $\fchanlinktwo$ is given
in Figure~\ref{figure:illustration_for_chan2_link}.
We point out
that Lemma~\ref{lemma:DFreducibility:4} still holds after
$\fchanlinktwo$ is added, and thus
Lemma~\ref{lemma:langch:deadlock-freedom} still holds as well.

A dyadic i-session involves only two roles: 0 and 1.
If we just study dyadic g-sessions corresponding to dyadic
i-sessions, then a channel type is of the form $\tchan(G,\ST)$ for
either $G=\iset{0}$ or $G=\iset{1}$. In this context, it is
unclear how Theorem~\ref{theorem:chan2_link} can be generalized.
When more than two roles are involved, there turns out to be a natural
generalization of Theorem~\ref{theorem:chan2_link}, which we report
in the next section.

\begin{figure}[htp]
\centering{%
\includegraphics[scale=0.5]{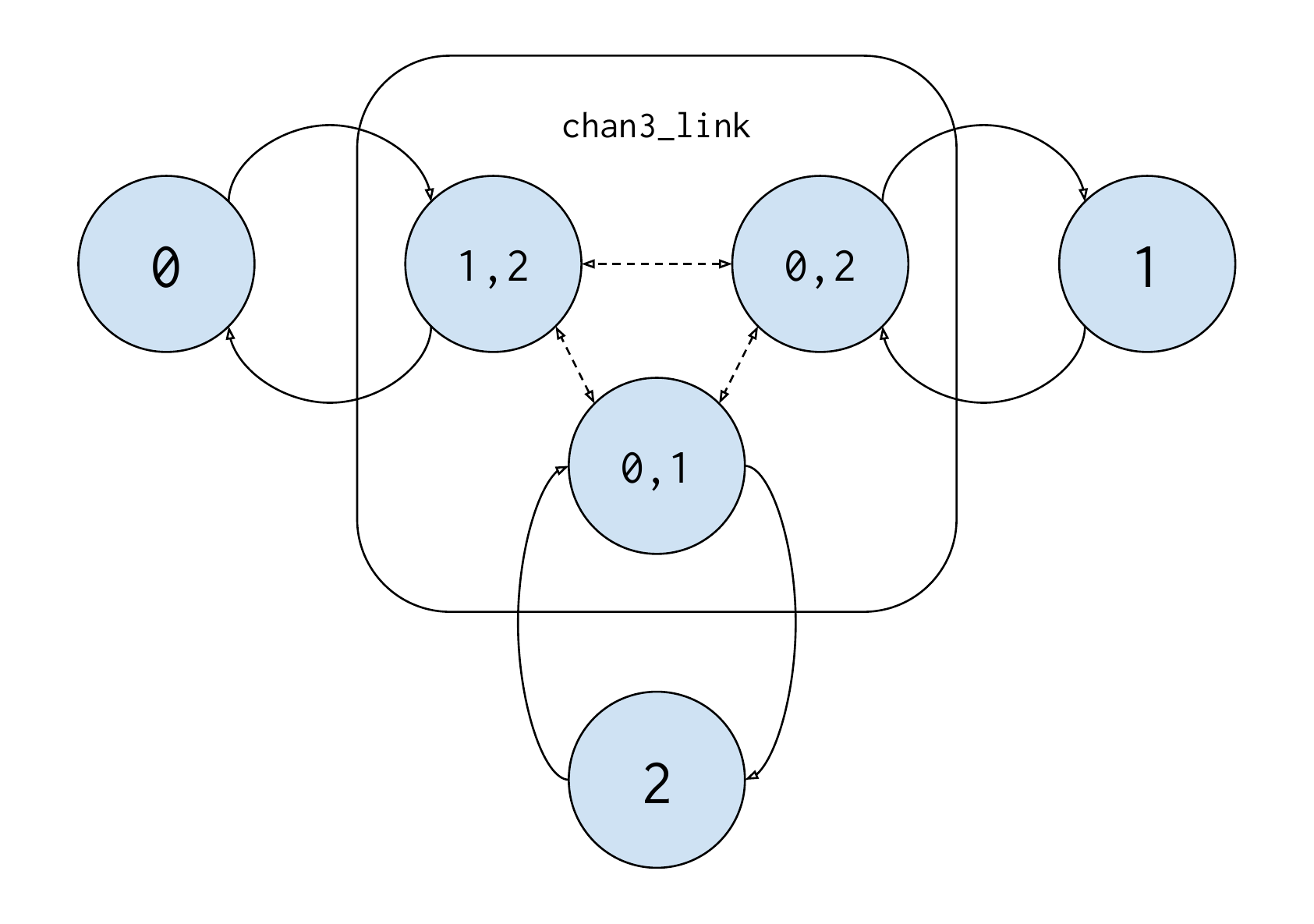}%
}
\caption{Illustrating $\fchanlinkthree$ involving 3 roles}
\label{figure:illustration_for_chan3_link}
\end{figure}
\subsection%
{Bidirectional Forwarding between Three Parties}
The next theorem states the existence of a generic method for
forwarding messages between three channels of certain types:
\begin%
{theorem}
\label{theorem:chan3_link}
Assume that $\chcst_0$ and $\chcst_1$ are two channels of the
types $\tchan(G_0,\ST)$ and $\tchan(G_1,\ST)$, respectively, where
$\setcomp{G_0}\cap\setcomp{G_1}=\emptyset$ holds, and $\chcst_2$ is
another channel of the type $\tchan(G_2,\ST)$ for
$G_2=\setcomp{G_0}\cup\setcomp{G_1}$.  Then there is a generic method
for forwarding each message received from one of $\chcst_0$,
$\chcst_1$ and $\chcst_2$ onto one of the other two in a type-correct
manner.
Let us use the name $\fchanlinkthree$ for a function of the following
type that implements this generic method:
$$
(\tchan(G_0,\ST),\tchan(G_1,\ST), \tchan(\setcomp{G_0}\cup\setcomp{G_1},\ST))\ctimp\tunit
$$
\end{theorem}
\begin%
{proof}
If $\ST$ is $\stnil$, then all that is needed is to call $\fclose$ on
each of $\chcst_0$, $\chcst_1$, and $\chcst_2$.  Assume $\ST$ is of
the form $\chmsg{i,j}{\ST_1}$ for some roles $i$ and $j$.  Note that
$\setcomp{G_0}$, $\setcomp{G_1}$, and $\setcomp{G_2}=G_0\cap G_1$ are
pairwise disjoint, and the union of these three equals the full set of
roles. So we have 9 scenarios covering all of the possibilities of
$i\in G$ and $j\in G'$ for $G$ and $G'$ ranging over $\setcomp{G_0}$,
$\setcomp{G_1}$, and $\setcomp{G_2}$.
\begin%
{itemize}
\item
Assume $i\in\setcomp{G_0}$ and $j\in\setcomp{G_0}$.
Then $\stmsg(i,j)$ indicates an external message for $\chcst_0$,
an internal message for $\chcst_1$, and
an internal message for $\chcst_2$. So this case is handled by
calling $\fskip$ on each of $\chcst_0$, $\chcst_1$, and $\chcst_2$.
\item
Assume $i\in\setcomp{G_0}$ and $j\in\setcomp{G_1}$.  Then
$\stmsg(i,j)$ indicates receiving for $\chcst_0$, sending for
$\chcst_1$ and an internal message for $\chcst_2$. So this case
is handled by calling $\frecv$ on $\chcst_0$ to receive a message, and
then calling $\fsend$ on $\chcst_1$ to send the message, and then
calling $\fskip$ on $\chcst_2$.
\item
Assume $i\in\setcomp{G_0}$ and $j\in \setcomp{G_2}$.  Then
$\stmsg(i,j)$ indicates receiving for $\chcst_0$, an internal message
for $\chcst_1$, and sending for $\chcst_2$.  So this case is handled
by calling $\frecv$ on $\chcst_0$ to receive a message, and then
calling $\fsend$ on $\chcst_2$ to send the message, and then calling
$\fskip$ on $\chcst_0$.
\item
Assume $i\in\setcomp{G_1}$ and $j\in\setcomp{G_0}$.  The case is
similar to the previous one where $i\in\setcomp{G_0}$ and
$j\in\setcomp{G_1}$ holds.
\item
Assume $i\in\setcomp{G_1}$ and $j\in\setcomp{G_1}$.  The case is
similar to the previous one where $i\in\setcomp{G_0}$ and
$j\in\setcomp{G_0}$ holds.
\item
Assume $i\in\setcomp{G_1}$ and $j\in \setcomp{G_2}$.  The case is similar to the
previous one where $i\in\setcomp{G_0}$ and $j\in \setcomp{G_2}$ holds.
\item
Assume $i\in \setcomp{G_2}$ and $j\in\setcomp{G_0}$.
Then $\stmsg(i,j)$ indicates sending for $\chcst_0$,
an internal message for $\chcst_1$, and receiving for
$\chcst_2$. So this case is handled by
calling $\frecv$ on $\chcst_2$ to receive a message,
and then calling $\fsend$ on $\chcst_0$ to send the message,
and then calling $\fskip$ on $\chcst_1$.
\item
Assume $i\in \setcomp{G_2}$ and $j\in\setcomp{G_1}$. The case is similar
to the previous one where $i\in\setcomp{G_1}$ and
$j\in\setcomp{G_2}$ holds.
\item
Assume $i\in \setcomp{G_2}$ and $j\in \setcomp{G_2}$.
Then $\stmsg(i,j)$ indicates an internal message for $\chcst_0$,
an internal message for $\chcst_1$, and
an external message for $\chcst_2$. So this case is handled by
calling $\fskip$ on each of $\chcst_0$, $\chcst_1$, and $\chcst_2$.
\end{itemize}
After one of the above 9 possibilities is performed, a recursive
call can be made on $\chcst_0$, $\chcst_1$ and $\chcst_2$ to perform
the rest of bidirectional forwarding between these channels.
\end{proof}
An illustration of $\fchanlinkthree$ involving 3 roles is given
in Figure~\ref{figure:illustration_for_chan3_link}.
We point out
that Lemma~\ref{lemma:DFreducibility:4} still holds after
$\fchanlinkthree$ is added, and thus
Lemma~\ref{lemma:langch:deadlock-freedom} still holds as well.

The following corollary
(which is partly stated as
Theorem~\ref{theorem:chan2_link_intro} in Section~\ref{section:introduction})
follows from Theorem~\ref{theorem:chan3_link} immediately:
\begin%
{corollary}
\label{theorem:chan2_link_create}
Assume that $\chcst_0$ and $\chcst_1$ are two channels of
the types $\tchan(G_0,\ST)$ and $\tchan(G_1,\ST)$, respectively, where
$\setcomp{G_0}\cap\setcomp{G_1}=\emptyset$ holds. Then there is a method
for creating a channel
$\chcst_2$ of the type $\tchan(G_0\cap G_1,\ST)$ such
that the generic method for forwarding messages as is stated in
Theorem~\ref{theorem:chan3_link} applies to $\chcst_0$, $\chcst_1$ and the
dual $\chcst'_2$ of $\chcst_2$.
Let us use the name $\fchanlinkcreate$ for a function of the
following type that implements the method for creating $\chcst_2$ based on
$\chcst_0$ and $\chcst_1$:
$$
(\tchan(G_0,\ST),\tchan(G_1,\ST))\ctimp\tchan(G_0\cap G_1,\ST)
$$
\end{corollary}
\begin%
{proof}
The channel $\chcst_2$ can simply
be obtained by evaluating the following expression:
$$\fchancreate(\lam{x}{\fchanlinkthree(\chcst_0, \chcst_1, x)})$$
\end{proof}

The very significance of Theorem~\ref{theorem:chan3_link} lies in its
establishing a foundation for a type-based approach to building
multiparty sessions based on dyadic g-sessions. Let us elaborate this
point a bit further. Suppose that we start with one dyadic session
where the two parties communicating via the dual channels $\chcst_0$
and $\chcst'_0$ and another dyadic session where the two parties
communicating via the dual channels $\chcst_1$ and $\chcst'_1$. If the
party holding both $\chcst'_0$ and $\chcst_1$ (which means the party
is shared between the two sessions) does bidirectional forwarding of
messages between them according to Theorem~\ref{theorem:chan2_link}
(that is, calling $\fchanlinktwo$ on $\chcst'_0$ and $\chcst_1$), then
a new session is created but it is still a dyadic one (where the
communication is between $\chcst_0$ and $\chcst'_1$). It is impossible
to build multiparty sessions (involving more than 2 parties) by simply
relying on Theorem~\ref{theorem:chan2_link} if we can only start with
dyadic ones. With Theorem~\ref{theorem:chan3_link}, three dyadic
sessions can be joined together by making a call to $\fchanlinkthree$,
resulting in the creation of a 3-party session.  In other words,
$\fchanlinkthree$ can be seen as a breakthrough.  In order to build
sessions involving more parties, we simply make more calls to
$\fchanlinkthree$.

\subsection%
{A Sketch for Building a 3-Party Session}
\label{subsection:sketch:3-party_session}
\def\tservice{{\bf service}}
\def\fservicecreate{\mbox{\it service\_create}}
\def\fservicerequest{\mbox{\it service\_request}}
Building a session often requires explicit coordination between the
involved parties during the phase of setting-up.  In practice,
designing and implementing coordination between 3 or more parties is
generally considered a difficult issue. By building a multiparty
session based on dyadic g-sessions, we only need to be concerned with
two-party coordination, which is usually much easier to handle.

Given a group $G$ of roles and a session type $\ST$, we introduce a
type $\tservice(G,\ST)$ that can be assigned to a value representing
some form of
{\it persistent} service.  With such a service,
channels of the type $\tchan(\setcomp{G},\ST)$ can be created
repeatedly. A built-in function $\fservicecreate$ is assigned the
following type for creating a service:
$$
(\tchan(G,\ST)\itimp\tunit)\ctimp\tservice(G,\ST)
$$
In contrast with $\fchancreate$ for creating a channel, $\fservicecreate$
requires that its argument be a non-linear function (so that this function
can be called repeatedly).
A client may call the following function to obtain a channel
to communicate with a server that provides the requested service:
\[\begin%
{array}{rcl}
\fservicerequest & : & (\tservice(G,\ST))\ctimp\tchan(\setcomp{G},\ST) \\
\end{array}\]
Suppose we want to build a 3-party session involving 3 roles: 0, 1,
and 2.  We may assume that there are two services of the types
$\tservice(\iset{0},\ST)$ and $\tservice(\iset{1},\ST)$ available to a
party (planning to implement role 2); this party can call
$\fservicerequest$ on the two services (which are just two names) to
obtain two channels $\chcst_0$ and $\chcst_1$ of the types
$\tchan(\iset{1,2},\ST)$ and $\tchan(\iset{0,2},\ST)$, respectively;
it then calls $\fchanlinkcreate(\chcst_0,\chcst_1)$ to obtain a
channel $\chcst_2$ of the type $\tchan(\iset{2},\ST)$ for
communicating with two servers providing the requested services.
Obviously, there are many other ways of building a multiparty session
by passing around multirole channels for dyadic g-sessions.

\def\tchoose{\mbox{\bf choose}}
\def\tchoosetag{\mbox{\bf ctag}}
\def\fchanchoosel{\mbox{\it chan\_choose\_l}}
\def\fchanchooser{\mbox{\it chan\_choose\_r}}
\def\fchoosetagl{\mbox{\it ctag\_l}}
\def\fchoosetagr{\mbox{\it ctag\_r}}
\def\fchanchoosetag{\mbox{\it chan\_choose\_tag}}
\def\tappend{\mbox{\bf append}}
\def\fchanappend{\mbox{\it chan\_append}}
\def\trepeat{\mbox{\bf repeat}}
\section{More Constructors for Session Types}
\label{section:more_st_constructors}
We present in this section
a few additional constructors for session types plus
an example of user-defined session type.

\subsection{Branching}
Given an integer $i$ (representing a role) and two session types
$\ST_0$ and $\ST_1$, we can form a branching session type $\tchoose(i,
\ST_0, \ST_1)$ that is given the following interpretation based a
group $G$ of roles:
\begin%
{itemize}
\item
Assume $i\in G$. Then the session type $\tchoose(i, \ST_0, \ST_1)$
means for the party implementing $G$ to send a message to the party
implementing $\setcomp{G}$ so as to inform the latter which of $\ST_0$ and
$\ST_1$ is chosen for specifying the subsequent communication.  In
order for Theorem~\ref{theorem:chan3_link} to work out in the presence
of $\tchoose$, this message needs to be sent repeatedly, targeting a
new role in $\setcomp{G}$ each time. It needs to be sent $n$ times
if there are $n$ roles in $\setcomp{G}$.
\item
Assume $i\not\in G$. Then the session type $\tchoose(i, \ST_0, \ST_1)$
means for the party implementing $G$ to wait for a message indicating
which of $\ST_0$ and $\ST_1$ is chosen for specifying the subsequent
communication. Note that this message arrives repeatedly for $n$
times, where $n$ is the cardinality of $G$.
\end{itemize}
As can be expected, we have two functions of the following type schemas:
$$
\begin%
{array}{lcr}
\fchanchoosel&:&(\tchan(G,\tchoose(i, \ST_0, \ST_1)))\ctimp\tchan(G,\ST_0) \\
\fchanchooser&:&(\tchan(G,\tchoose(i, \ST_0, \ST_1)))\ctimp\tchan(G,\ST_1) \\
\end{array}
$$
where $i\in G$ is assumed. We have another function $\fchanchoosetag$ of the
following type schema:
$$
\begin%
{array}{lcr}
(\tchan(G,\tchoose(i, \ST_0, \ST_1)))\ctimp\exists\sigma.~(\tchoosetag(\ST_0,\ST_1,\sigma),\tchan(G,\sigma)) \\
\end{array}$$
where $\tchoosetag$ is a datatype with the following two constructors
(which are essentially represented as 0 and 1):
$$
\begin%
{array}{lcr}
\fchoosetagl &:& ()\ctimp \tchoosetag(\ST_0, \ST_1, \ST_0) \\
\fchoosetagr &:& ()\ctimp \tchoosetag(\ST_1, \ST_0, \ST_1) \\
\end{array}$$
By performing pattern matching on the first component of the tuple
returned by a call to $\fchanchoosetag$, one can tell whether the
$\sigma$ is $\ST_0$ or $\ST_1$.
Note that the existential quantification is available in ATS but it is
not part of $\langch$. Hopefully, the idea should be clear to the reader.

\subsection{Sequencing}
Given two session types $\ST_0$ and $\ST_1$, we can form another
one of the form $\tappend(\ST_0, \ST_1)$, which intuitively
means the standard sequencing of $\ST_0$ and $\ST_1$. Also, we have
a function $\fchanappend$ of the following type schema for operating
on a channel of some sequencing session type:
\[
\fontsize{8pt}{9pt}\selectfont
\begin%
{array}{c}
(\tchan(G, \tappend(\ST_0, \ST_1)), \tchan(G, \ST_0)\ltimp\tchan(G, \stnil))\ctimp\tchan(G, \ST_1) \\
\end{array}\]
When applied to a channel and a linear function, $\fchanappend$
essentially calls the linear function on the channel to return another
channel. There is a bit of cheating here because there is no
type-based enforcement of the requirement that the channel returned by
the linear function be the same as the one passed to it, which on the
other hand can be readily achieved in ATS.

\subsection%
{Repeating Indefinitely}
Given an integer $i$ (representing a role) and a session type $\ST$, we
can form another session type of the form $\trepeat(i, \ST)$, which is
essentially defined recursively as follows:
$$
\trepeat(i, \ST) = \tchoose(i, \stnil, \trepeat(i, \ST))
$$
Intuitively, $\trepeat(i, \ST)$ means that the party implementing
the role $i$ determines whether a session specified by $\ST$ should
be repeated indefinitely. For instance, a list-session (between two
parties) can be specified as follows:
$$\trepeat(0, \chmsg{0,1,T}{\stnil})$$
which means the server (implementing the role $0$) offers to the client
(implementing the role $1$) a list of values of the type $T$. Similarly,
a colist-session (between two
parties) can be specified as follows:
$$\trepeat(1, \chmsg{0,1,T}{\stnil})$$ which means the client requests
from the server a list of values of the type $T$.

\begin%
{figure}
\fontsize{8pt}{9pt}\selectfont
\begin%
{verbatim}
datatype
ssn_queue_(a:vtype, i:int, int) =
| queue_nil(a,i,0) of (nil)
| {n:nat}
  queue_enq(a,i,n) of msg(i, 0, a)::ssn_queue(a, n+1)
| {n:pos}
  queue_deq(a,i,n) of msg(0, i, a)::ssn_queue(a, n-1)

where ssn_queue(a:vtype, n:int) =
  choose(0, ssn_queue_(a, 1, n), ssn_queue_(a, 2, n))
\end{verbatim}
\caption{A type for queue-sessions}
\label{figure:ssn_queue}
\end{figure}
\subsection%
{User-Defined Session Types}
We can readily make use of various advanced programming features in
ATS for formulating types for sessions as well as implementing these
sessions.  As a concrete example, a type for queue-sessions written in
the source syntax of ATS is given in Figure~\ref{figure:ssn_queue}, which
makes direct use of dependent types (of DML-style).  Given the
following explanation, we reasonably expect that the reader be able to
make sense of this interesting example.

During a queue-session (specified by $\mbox{\it ssn\_queue}$), there
are one server(S0) and two clients (C1 and C2); the server chooses
(based on some external information) which client is to be served in
the next round; the chosen client can request the server to create an
empty queue (\verb|queue_nil|), enqueue an element into the current
queue (\verb|queue_enq|), and dequeue an element from the current
non-empty queue (\verb|queue_deq|). Given a type $T$ and a natural
number $N$, the type $\mbox{\it ssn\_queue}(T, N)$ means that the size
of the underlying queue in the current queue-session is $N$.  This
example is largely based on a type for 2-party queue-sessions in
SILL~\cite{SILL}.  What is novel here is an added party plus the use
of dependent types (of DML-style) to specify the size of the
underlying queue in a queue-session.
As a side note, one may be wondering why C1 can keep the correct
account of queue length after C2 performs an operation. The very
reason is that C2 announces to both S0 and C1 which operation C2 is to
perform before it performs it.

\begin%
{figure}
\fontsize{8pt}{9pt}\selectfont
\begin%
{verbatim}
typedef
ssn_s0b1b2_fail = nil

typedef
ssn_s0b1b2_succ =
msg(B2,S0,proof) ::
msg(S0,B2,receipt) :: nil

typedef
ssn_s0b1b2 =
msg(B1,S0,title) :: msg(S0,B1,price) ::
msg(S0,B2,price) :: msg(B1,B2,price) ::
choose(B2,ssn_s0b1b2_succ,ssn_s0b1b2_fail)
\end{verbatim}
\caption{A classic 3-role session type}
\label{figure:ssn_s0b1b2}
\end{figure}
\section%
{An Example of 3-Party Session}
\label{section:S0B1B2}
Let us assume the availability of three roles: Seller(S0), Buyer~1(B1)
and Buyer~2(B2).  A description of the classic
one-seller-and-two-buyers(S0B1B2) protocol due to (Honda et al. 2008) is
essentially given as follows:
\begin{enumerate}
\item B1 sends a book title to S0.
\item S0 replies a quote to both Buyer 1 and Buyer 2.
\item B1 tells B2 how much B1 can contribute:
\begin{enumerate}
\item Assume B2 can afford the remaining part:
\begin{enumerate}
\item B2 sends S0 a proof of payment. 
\item S0 sends B2 a receipt for the sale.
\end{enumerate}
\item Assume B2 cannot afford the remaining part:
\begin{enumerate}
\item B2 terminates.
\end{enumerate}
\end{enumerate}
\end{enumerate}
This protocol is formally captured by the type $\mbox{\it ssn\_s0b1b2}$
given in Figure~\ref{figure:ssn_s0b1b2}. For taking a peek at a running
implementation of this example in ATS, please visit the following link:
\begin%
{center}
\texttt{http://pastebin.com/JmZRukRi}
\end{center}
The steps involved in building a 3-party session
are basically those outlined in Section~\ref{subsection:sketch:3-party_session}.

\section%
{Implementing Session-Typed Channels}
\label{section:implementation}
As far as implementation is of the concern, there is very little that
needs to be done regarding typechecking in order to support
session-typed channels in ATS. The only considerably significant
complication comes from the need for solving constraints generated
during typechecking that may involve various common set operations (on
groups of roles), which the current built-in constraint-solver for ATS
cannot handle. Fortunately, we have an option to export such
constraints for them to be solved with an external constraint-solver
based on Z3~\cite{Z3-tacas08}.

The first implementation of session-typed channels (based on shared
memory) for use in ATS is done in ATS itself, which compiles to C, the
primary compilation target for ATS. Another implementation of
session-typed channels (based on processes) is done in Erlang. As the
ML-like core of ATS can already be compiled into Erlang, we have now an
option to construct distributed programs in ATS that may make use of
session types and then translate these programs into Erlang code for
execution, thus taking great advantage of the infrastructural support
for distributed computing in Erlang.

We outline some key steps taken in both of the implementations. In
particular, we briefly mention an approach to implementing
$\fchanlinktwo$ and $\fchanlinkthree$ that completely removes the need
for explicit forwarding of messages and thus the inefficiency associated
with it.

\def\uch{\mbox{\it uch}}
Let us use $\uch$ to refer to a uni-directional channel that can be
held by two parties; one party can only write to it while the other
can only read from it. Suppose we want to build a pair of multirole
channels $\chpcst$ and $\chncst$ for some groups $G$ and $\setcomp{G}$
of roles; we first create a matrix $M$ of the dimension $\nrole$ by
$\nrole$;, where $\nrole$ is the total number of available roles; for
each $i\in G$ and $j\in\setcomp{G}$, we use $\uch(i,j)$ and
$\uch(j,i)$ to refer to the two uni-directional channels stored in
$M[i,j]$ and $M[j,i]$, respectively. Note that this matrix $M$ is
shared by both $\chpcst$ and $\chncst$; for $\chpcst$, $\uch(i,j)$ and
$\uch(j,i)$ are used to send messages from role $i$ to role $j$ and
receive messages sent from role $j$ to role $i$ on $\chncst$,
respectively; for $\chncst$, it is precisely the opposite.

If $\fchanlinktwo$ is implemented by following
Theorem~\ref{theorem:chan2_link} directly, then a call to
$\fchanlinktwo$ creates a thread/process for performing bidirectional
forwarding of messages explicitly, and the created thread/process only
terminates after no more forwarding is needed. If we assume that a
$\uch$ can be sent onto another $\uch$, which can be readily supported
in both ATS and Erlang, then a much more efficient approach to
implementing $\fchanlinktwo$ can be described as follows. Suppose we
call $\fchanlinktwo$ on two channels $\chcst_0$ and $\chcst_1$ of the
types $\tchan(G,\ST)$ and $\tchan(\setcomp{G},\ST)$; let $M_0$ and
$M_1$ be the matrices in $\chcst_0$ and $\chcst_1$ for holding
uni-directional channels in $\chcst_0$ and $\chcst_1$, respectively;
for each $i\in G$ and $j\in\setcomp{G}$, we send the $\uch$ in
$M_1[i,j]$ onto the one in $M_0[i,j]$ and the $\uch$ in $M_0[j,i]$
onto the one in $M_1[j,i]$; if a $\uch$ is received on the one in
$M_0[i,j]$ ($M_1[j,i]$), then the $\uch$ is put into $M_0[i,j]$
($M_1[j,i]$) to replace the original one. It should be clear that
$\fchanlinkthree$ can implemented similarly.

\section%
{Related Work and Conclusion}
\label{section:related-conclusion}

Session types were introduced by Honda~\cite{Honda93} and further
extended subsequently~\cite{TakeuchiHK94,HondaVK98}.  There have since
been extensive theoretical studies on session types in the
literature(e.g.,~\cite{CastagnaDGP09,GayV10,CairesP10,ToninhoCP11,Vasconcelos12,Wadler12,LindleyM15}).
Multiparty session types, as a generalization of (dyadic) session
types, were introduced by Honda and others~\cite{HondaYC08}, together
with the notion of global types, local types, projection and
coherence.  By introducing dyadic group sessions (g-sessions), we
give a novel form of generalization going from dyadic sessions to
dyadic g-sessions.

\def\tchpos{\mbox{\bf chpos}}
\def\tchneg{\mbox{\bf chneg}}
The notion of dyadic g-sessions is rooted in a very recent attempt
to incorporate session types for dyadic sessions into
ATS~\cite{MTLC-i-sessiontype}. In an effort to formalize session types,
two kinds of channel types $\tchpos(\ST)$ and $\tchneg(\ST)$ are
introduced for positive and negative channels, respectively, directly
leading to the discovery of the notion of single-role channels (as
$\tchpos(\ST)$ and $\tchneg(\ST)$ can simply be translated into
$\tchan(0,\ST)$ and $\tchan(1,\ST)$, respectively) and then the
discovery of the key notion of multirole channels in this paper.

In~\cite{Denielou:2011gl,Neykova:2014ib}, a party can play multiple
roles by holding channels belonging to multiple sessions. We see no
direct relation between such a multirole party and a multirole
channel.  In~\cite{Carbone:2015hl}, coherence is treated as a
generalization of duality. In particular, the binary cut rule is
extended to a multiparty cut rule. But this multiparty cut rule is not
directly related to Theorem~\ref{theorem:chan3_link} as far as we can
tell.




It is in general a challenging issue to establish deadlock-freedom for
session-typed concurrency. There are variations of session types that
introduce a partial order on time stamps~\cite{SumiiK98} or a
constraint on dependency graphs \cite{abs-1010-5566}.  As for
formulations of session types (e.g.,~\cite{CairesP10,Wadler12}) based
on linear logic~\cite{GIRARD87}, the standard technique for
cut-elimination is often employed to establish global progress (which
implies deadlock-freedom). In $\langch$, there is no explicit tracking
of cut-rule applications in the type derivation of a program.  The
notion of DF-reducibility (taken from~\cite{MTLC-i-sessiontype}) is
introduced in order to carry out cut-elimination in the absence of
explicit tracking of cut-rule applications.

Probably, $\langch$ is most closely related to
SILL~\cite{ToninhoCP13}, a functional programming language that adopts
via a contextual monad a computational interpretation of linear
sequent calculus as session-typed processes.  Unlike in $\langch$, the
support for linear types in SILL is not direct and only monadic values
(representing open process expressions) in SILL can be linear.  In
terms of theoretical development, the approach to establishing global
progress in $\langch$ is rooted in the one for SILL (though the latter
does not apply directly).

Also, $\langch$ is related to previous work on incorporating session
types into a multi-threaded functional
language~\cite{VasconcelosRG04}, where a type safety theorem is
established to ensure that the evaluation of a well-typed program can
never lead to a so-called {\em faulty configuration}. However, this
theorem does not imply global progress as a program that is not of
faulty configuration can still deadlock.  Also, we point out $\langch$
is related to recent work on assigning an operational semantics to a
variant of GV~\cite{LindleyM15}. In particular, the approach based on
DF-reducibility to establishing global progress in $\langch$ is
analogous to the one taken to establish deadlock-freedom for this
variant.

As for future work we are particularly interested in applying the
notion of dyadic g-sessions to the design and formalization of a type
system for some variant of $\pi$-calculus. Also, it should be both
exciting and satisfying if a logic-based interpretation can be found
for Theorem~\ref{theorem:chan3_link}.


There are a variety of programming issues that need to be
addressed in order to facilitate the use of session types in
practice. Currently, session types are often represented as
datatypes in ATS, and programming with such session types tends to
involve writing a very significant amount of boilerplate code. In the
presence of large and complex session types, writing such code can be
tedious and error-prone. Naturally, we are interested in developing
some meta-programming support for generating such code automatically.
Also, we are in the process of designing and implementing session
combinators (in a spirit similar to parsing
combinators~\cite{Hutton92}) that can be conveniently called to
assemble subsessions into a coherent whole.


\end{document}